\documentclass[12pt,english]{article}
\usepackage{amsmath}
\usepackage{amsfonts}
\usepackage{amssymb}
\usepackage{amsthm}

\usepackage[left=25mm,top=26mm,right=26mm,bottom=30mm]{geometry}

\usepackage{enumerate}
\usepackage{graphicx}

\newtheorem{theorem}{Theorem}[section]

\newtheorem{lemma}{Lemma}[section]
\newtheorem{proposition}{Proposition}[section]

\let\epsilon=\varepsilon
\def\tire{\thinspace--\thinspace}

\newcommand{\eps}{\varepsilon}

\def\notes#1{}

\title{How Smooth  Should be  the System Initially to Escape Unbounded Chaos}

\author{A.A.~Lykov\thanks{Mechanics and Mathematics Faculty, Lomonosov Moscow
State University, Leninskie Gory~1, Moscow, 119991, Russia} \and  V.A.~Malyshev\footnotemark[1]}

\begin{document}
\maketitle

\begin{abstract}
We consider infinite harmonic chain on the real line with deterministic
dynamics (no stochasticity). We indicate classes of uniformly bounded
initial conditions when the trajectories of particles stay uniformly
bounded. 
\end{abstract}

\tableofcontents{}
\allowdisplaybreaks

\section{The model}

This paper is, in some sense, a continuation of our paper \cite{LM_1},
and uses the same model and notation. Thus, we consider infinite particle
system with trajectories $\{x_{k}(t),k\in Z\}$ on $R$. Put $q_{k}=x_{k}-a_{k}$
for fixed 
\[
\ldots <a_{k}<a_{k+1}<\ldots 
\]
The formal potential energy is defined as 
\[
U=\frac{\omega^{2}}{2}\sum_{k}(x_{k+1}-x_{k}-a_{k+1}+a_{k})^{2}=\frac{\omega^{2}}{2}\sum_{k}(q_{k+1}-q_{k})^{2}
\]
where $\omega>0$. The trajectories are defined by the following linear
system of equations for $q_{k}(t)$ 
\begin{equation}
\frac{d^{2}q_{k}}{dt^{2}}=\omega^{2}(q_{k+1}-2q_{k}+q_{k-1})=\omega^{2}(\Delta q)_{k}\label{infSystem}
\end{equation}
and initial conditions are always assumed to be $q(0)\in l_{\infty},\ p(0)=0$ (we denote $p(t)=\dot{q}(t)$),
that is, 
\[
\sup_{k}|q_{k}(0)|<\infty,\ p_{k}(0)=0.
\]
We will try to find sub-classes of the initial conditions such that $q_{k}(t)$
are bounded uniformly in $k$ and $t$.

\section{Results}

We will need the following definitions. For any sequence $q\in l_{\infty}$
define the new sequence 
\[
q^{\Delta}=-\Delta q,\quad q_{k}^{\Delta}=2q_{k}-q_{k+1}-q_{k-1},\ k\in\mathbb{Z}.
\]
Denote by $l^{\Delta}\subset l_{\infty}(\mathbb{Z})$ the set of sequences
$q\in l_{\infty}(\mathbb{Z})$, for which the following conditions
hold: 
\begin{enumerate}
\item $q^{\Delta}\in l_{2}(\mathbb{Z})$. In this case the Fourier transform
of $q^{\Delta}$ is defined as 
\[
Q^{\Delta}(\lambda)=\sum_{k}e^{ik\lambda}q_{k}^{\Delta}\in L_{2}([0,2\pi]).
\]
\item For some real number $A\in\mathbb{R}$ the function 
\begin{equation}
h(\lambda)=\frac{1}{\lambda}\left(\frac{Q^{\Delta}(\lambda)}{\lambda}-iA\right)\label{phiDef_s}
\end{equation}
belongs to $L_{1}[0,\pi]$, where $i^{2}=-1$, that is 
\[
\int_{0}^{\pi}| h(\lambda)|d\lambda<\infty.
\]
\end{enumerate}

Main theorems, which are given below, about uniform boundedness and limit behavior, will
hold for initial conditions from subspace $l^{\Delta}$.

It is clear that $l^{\Delta}$ is linear space over $\mathbb{R}$,
and below we shall present its properties in more detail. But first
of all we will explain the intuitive sense of condition (\ref{phiDef_s}):
we show informally that (\ref{phiDef_s}) holds if the sequence $q_{k}^{\Delta}$
tends to zero sufficiently fast if $|k|\to\infty$. It is clear 
that $h(\lambda)$ is absolutely integrable on any interval $[\delta,\pi],\ \delta>0$.
That is why we should understand what occurs around zero. One can
write: 
\[
q_{k}^{\Delta}=-(\delta_{k+1}-\delta_{k}),\ \quad\delta_{k}=q_{k}-q_{k-1}.
\]
Thus it is natural to expect that $Q^{\Delta}(0)=0$. It will
follow that in the case when $Q^{\Delta}(\lambda)$ is sufficiently smooth
we could write 
\[
Q^{\Delta}(\lambda)=c\lambda+O(\lambda^{2}),\quad c=\frac{d}{d\lambda}Q^{\Delta}(0)=i\sum_{k}kq_{k}^{\Delta}.
\]
Consequently, if we put $A=\sum_{k}kq_{k}^{\Delta}$ in $(\ref{phiDef_s})$,
we will see that $h(\lambda)=O(1)$, and thus $h(\lambda)\in L_{1}([0,\pi])$,
and the corresponding condition on $h(\lambda)$ holds.

Note now that condition \ref{phiDef_s} is equivalent to the following one
(that we will use below): for some real number $A\in\mathbb{R}$ the
function 
\begin{equation}
\phi(\lambda)=\frac{1}{\sin (\lambda / 2)}\Bigl(\frac{Q^{\Delta}(\lambda)}{\sin (\lambda / 2)}-iA\Bigr)\label{phiDef}
\end{equation}
belongs to $L_{1}[0,\pi]$.

\begin{theorem}[On uniform boundedness] \label{unBoundTh} Assume
that $q(0)\in l^{\Delta},\ p(0)=0$, then the solution $q(t)$ is
uniformly bounded, that is, 
\[
\sup_{t\geqslant0}\sup_{k\in\mathbb{Z}}|q_{k}(t)|<\infty.
\]
\end{theorem}

\begin{theorem}[On large time behaviour of the system] \label{limitTheorem}
  $\phantom{aa}$  \linebreak Assume that 
  $q(0)\in l^{\Delta},$ $p(0)=0$, then there exists $\nu\in\mathbb{R}$
such that for any $k\in\mathbb{Z}$ we have: 
\[
\lim_{t\rightarrow\infty}q_{k}(t)=\nu.
\]
\end{theorem}

Number $\nu$ is also related to the limit of $q_{k}(0)$ when $k\to\infty$:
see theorem \ref{ldeltainfinity} below.

\subsection{Properties of the space $\boldsymbol{l^{\Delta}}$}

First, we give examples of sequences $q\in l^{\Delta}$. 

1. \ (Sign) Put 
\[
q_{k}=\mathrm{sign}(k)=\begin{cases}
1, & k>0,\\
0, & k=0,\\
-1, & k<0.
\end{cases}
\]
Obviously, $q_{k}^{\Delta}=0$ for $|k|>1$. Moreover, 
\[
q_{1}^{\Delta}=1,\quad q_{-1}^{\Delta}=-1,\quad q_{0}^{\Delta}=0.
\]
That is why 
\[
Q^{\Delta}(\lambda)=(e^{i\lambda}-e^{-i\lambda})=2i\sin(\lambda).
\]
Put $A=4$ in (\ref{phiDef}). Then 
\[
\phi(\lambda)=\frac{1}{\sin \lambda / 2}\Bigl(\frac{2\sin(\lambda)}{\sin \lambda / 2}-4\Bigr)=\frac{4}{\sin \lambda / 2}\Bigl(\cos \frac{\lambda}{ 2}-1\Bigr).
\]
It is clear that $\phi(\lambda)\in L_{1}[0,\pi]$. Then $\mathrm{sign}(k)\in l^{\Delta}$.
Below we give the graph of the solution with $\omega= 1/2$ and
initial condition $q_{k}(0)=\mathrm{sign}(k),\ p(0)=0$.

\begin{figure}[!hbtp]
  \hspace*{-20mm}\includegraphics[scale=0.33]{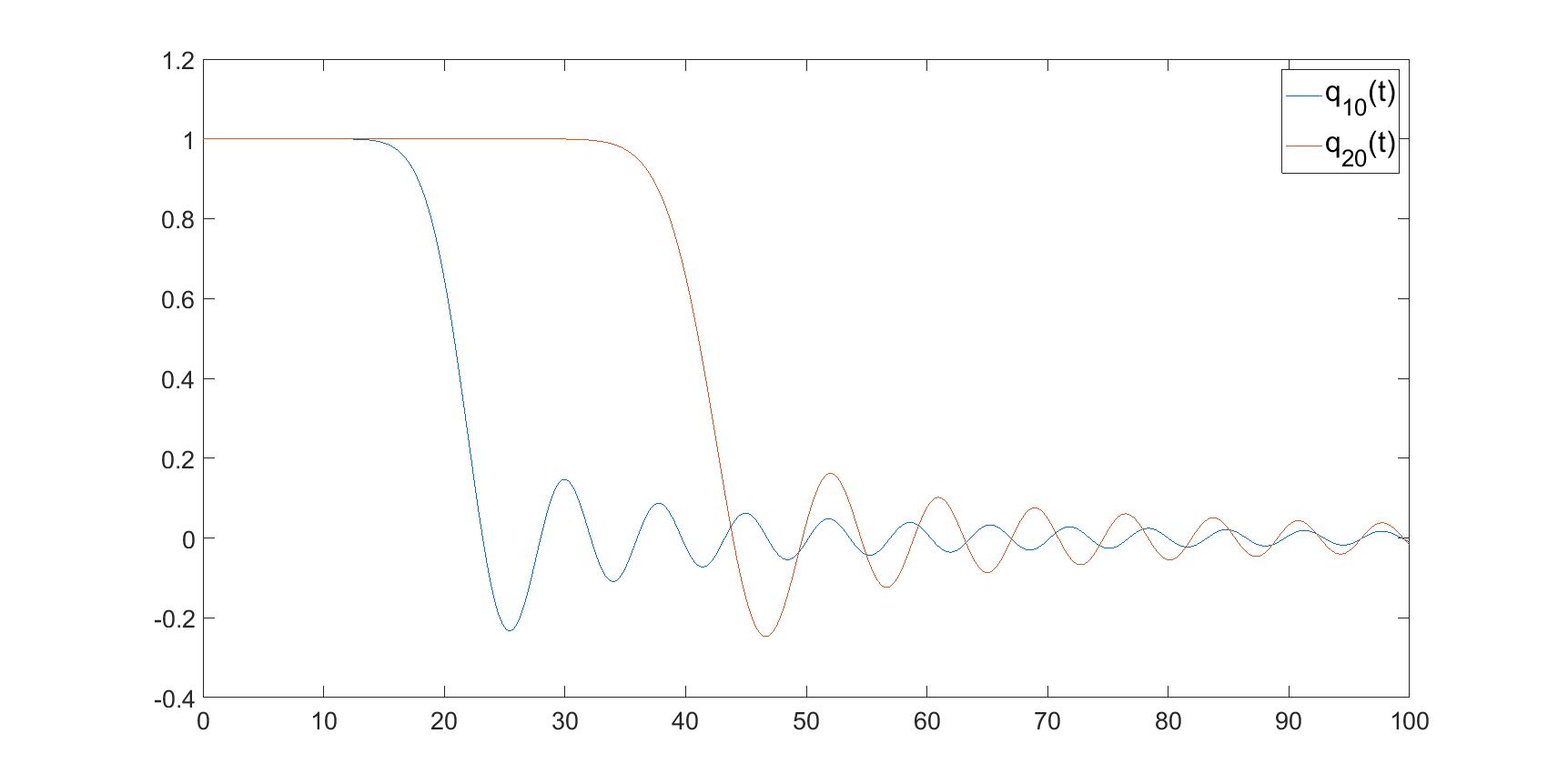}
\end{figure}
  
Both particles with numbers $10$ and $20$, up to the time of order
$t\ll 2n$ oscillate around point $1$ with exponentially small amplitude.
These oscillations are even not seen on the graph. Then they ``quickly''
enter the regime of damped oscillations around the equilibrium position.
In our case the solution is given by formula (3.2) of \cite{LM_1}:
\begin{equation}
q_{n}(t)=J_{0}(t)+2\sum_{k=1}^{n-1}J_{2k}(t)+J_{2n}(t)=1+J_{2n}(t)-2\sum_{k=2n}^{\infty}J_{2k}(t),\ n\geqslant1,\label{q_n_t}
\end{equation}
where 
\[
J_{n}(t)=\frac{1}{\pi}\int_{0}^{\pi}\cos(nx-t\sin x)dx,\ t\geqslant0
\]
is the Bessel function of the first kind. In (\ref{q_n_t}) we used
the known formula (see \cite{Watson}): 
\[
2\sum_{k=1}^{\infty}J_{2k}(t)+J_{0}(t)=1.
\]

\medskip

2. \ Now consider example which is in some sense opposite to the previous one:
\[
q_{k}=\begin{cases}
1, & k\ne0,\\
b, & k=0,
\end{cases}
\]
for some $b\in\mathbb{R}$. Then 
\begin{align*}
  Q^{\Delta}(\lambda)&=e^{i\lambda}(2-b-1)+2b-2+e^{-i\lambda}(2-b-1)\\
  &=2(b-1)(1-\cos\lambda)=4(b-1)\sin^{2}\frac{\lambda}{2}.
\end{align*}
Put $A=0$ in (\ref{phiDef}). Then 
\[
\phi(\lambda)=4(b-1).
\]
Again we see that $\phi(\lambda)\in L_{1}[0,\pi]$ and thus $q_{k}\in l^{\Delta}$.
Uniform boundedness for this case could be proven differently. Namely,
we have the following presentation: 
\[
q=g+(b-1)e_{0},
\]
where the sequence $g$ consists of one's only, $e_{0}$ contains
only zeroes except the zeroth component which is equal to $1$. Then $\Delta g=0$,
and hence the solution can be written as 
\[
q(t)=g+(b-1)\tilde{q}(t),
\]
where $\tilde{q}(t)$ is the solution with initial condition $e_{0}\in l_{2}$,
that is evidently uniformly bounded.

\medskip

3. \  Now consider the following sequence 
\[
q_{k}=(-1)^{k}.
\]
Then 
\[
(\Delta q)_{k}=(-1)^{k}(-1-1-2)=-4q_{k}
\]
and $q\notin l^{\Delta}$. Nevertheless one can prove uniform boundedness
with such initial condition. It is known that 
\[
q(t)=\cos(t\sqrt{V})q(0),\quad V=-\omega_{1}^{2}\Delta
\]
(see \cite{LM_1}, lemma 3.3). 
It follows that 
\[
q(t)=\sum_{k=0}^{\infty}(-1)^{k}\frac{t^{2k}V^{k}}{(2k)!}q=\sum_{k=0}^{\infty}(-1)^{k}\frac{(4\omega_{1}^{2})^{k}t^{2k}}{(2k)!}q=\cos(2\omega_{1}t)q.
\]
The uniform boundedness of $q(t)$ follows. 

\begin{theorem} \label{sufficientCondld} Assume that 
\begin{equation}
\sum_{k\ne0}|q_{k}^{\Delta}||k|\ln|k|<\infty . \label{qvCondforlv}
\end{equation}
Then $q\in l^{\Delta}$. \end{theorem}

For example, consider the sequence with 
\[
q_{k}=\frac{\sin(\ln\ln|k|)}{\ln^{2}(|k|)},\ \mbox{for}\ |k|>1
\]
and $q_{k}=0$ for $|k|\leqslant1$. It is easy to see that 
\[
q_{k}^{\Delta}=O\left(\frac{1}{k^{2}\ln^{3}|k|}\right).
\]
Then the conditions of theorem \ref{sufficientCondld} hold, and thus
$q\in l^{\Delta}$.

Generalizing the previous example we put 
\[
q_{k}=f(k)
\]
for some $C^{2}$-smooth and bounded function $f(x),\ x\in\mathbb{R}$.
Then simple arguments show the convergence of the integral 
\begin{equation}
\int_{-\infty}^{+\infty}|f''(x)| \, |x|\ln(1+|x)\ dx<\infty \label{202001261810}
\end{equation}
implies (\ref{qvCondforlv}) and $q\in l^{\Delta}$. Indeed, by Lagrange
theorem we can write 
\[
q_{k}^{\Delta}=f'(x_{k+1})-f'(x_{k})
\]
for some points $x_{k}\in(k-1,k)$. Putting $h(x)=|x|\ln(1+|x|)$,
we get the inequalities: 
\begin{align*}
  \sum_{k\ne0}|q_{k}^{\Delta}| \, |k|\, \ln|k|& \leqslant\sum_{k}|f'(x_{k+1})-f'(x_{k})|h(x_{k+1})\\
  &\leqslant\sum_{k}h(x_{k+1})\int_{x_{k}}^{x_{k+1}}|f''(x)|\ dx 
\leqslant2\int_{-\infty}^{+\infty}|f''(x)|h(x)dx.
\end{align*}
The latter inequality follows from the mean value theorem, because
for some point $u_{k}\in(x_{k},x_{k+1})$ the following equality holds:
\[
\int_{x_{k}}^{x_{k+1}}|f''(x)|h(x)\ dx=h(u_{k})\int_{x_{k}}^{x_{k+1}}|f''(x)|\ dx\geqslant\frac{1}{2}h(x_{k+1})\int_{x_{k}}^{x_{k+1}}|f''(x)|\ dx.
\]
Thereby we have proven that (\ref{qvCondforlv}) follows from (\ref{202001261810}).

\begin{theorem}[Limits at infinity of $\boldsymbol{l^\Delta}$] \label{ldeltainfinity}
Assume that $q\in l^{\Delta}$. Then the following finite limits exist
$
\lim_{k\rightarrow+\infty}q_{k}=L_{+},$ $\lim_{k\rightarrow-\infty}q_{k}=L_{-}
$
and moreover the following equalities hold: 
\begin{align}
L_{+}-L_{-}= & \ \frac{A}{2},\label{lplusminulimitminus}\\
\frac{L_{+}+L_{-}}{2}= & \ \nu,\label{lplusminulimitplus}
\end{align}
where number $A$ was defined in (\ref{phiDef}), and $\nu$ was introduced
in theorem \ref{limitTheorem}.
\end{theorem}

\section{Proof of Theorem \ref{ldeltainfinity}}

Consider the formal Fourier transform of $q_{k}$ 
\[
Q(\lambda)=\sum_{k}e^{ik\lambda}q_{k}
\]
and formal equalities following from it: 
\begin{align}
  Q^{\Delta}(\lambda)&=\sum_{k}e^{ik\lambda}q_{k}^{\Delta}=\sum_{k}e^{ik\lambda}(2q_{k}-q_{k-1}-q_{k+1}) \nonumber \\
  &=(2-2\cos\lambda)Q(\lambda)=4\sin^{2}\left(\frac{\lambda}{2}\right)Q(\lambda), \nonumber
\\
  q_{n}&=\frac{1}{2\pi}\int_{0}^{2\pi}e^{-in\lambda}Q(\lambda)d\lambda=\frac{1}{2\pi}\int_{0}^{2\pi}e^{-in\lambda}
         \frac{1}{4\sin^{2}(\lambda / 2 )}Q^{\Delta}(\lambda)d\lambda.\label{qkformalformula}
\end{align}
Then $Q^{\Delta}$ can be written as the sum of even and odd functions:
\begin{equation}
Q^{\Delta}(\lambda)=Q_{+}^{\Delta}(\lambda)+iQ_{-}^{\Delta}(\lambda),\label{QDeltaSymAsymF}
\end{equation}
where 
\[
Q_{+}^{\Delta}(\lambda)=q_{0}^{\Delta}+\sum_{k=1}^{\infty}q_{k}^{\Delta,+}\cos k\lambda,\quad Q_{-}^{\Delta}(\lambda)=\sum_{k=1}^{\infty}q_{k}^{\Delta,-}\sin k\lambda,\quad q_{k}^{\Delta,\pm}=q_{k}^{\Delta}\pm q_{-k}^{\Delta}.
\]
It is clear that 
\[
Q_{+}^{\Delta}(\lambda)=Q_{+}^{\Delta}(2\pi-\lambda),\quad Q_{-}^{\Delta}(\lambda)=-Q_{-}^{\Delta}(2\pi-\lambda).
\]
Substitute (\ref{QDeltaSymAsymF}) to the formula (\ref{qkformalformula})
for $q_{k}$ and note that $q_{k}$ is real. Again formally we get:
\[
  q_{n}=\frac{1}{2\pi}\int_{0}^{2\pi}(Q_{+}^{\Delta}(\lambda)\cos(n\lambda)+Q_{-}^{\Delta}(\lambda)\sin(n\lambda))
  \frac{1}{4\sin^{2}(\lambda / 2 )}\, d\lambda.
\]
Now using the symmetry of the integrand with respect to point $2\pi$,
we can rewrite the last formula as: 
\[
  q_{n}=\frac{1}{\pi}\int_{0}^{\pi}(Q_{+}^{\Delta}(\lambda)\cos(n\lambda)+Q_{-}^{\Delta}(\lambda)\sin(n\lambda))
  \frac{1}{4\sin^{2}(\lambda / 2)}\, d\lambda.
\]
It is not difficult to see that the function $\phi$, defined in the
second condition (\ref{phiDef}) in the definition of $l^{\Delta}$,
belongs to $L_{1}([0,\pi])$ if and only  if 
\[
  \phi^{+}(\lambda)=\frac{Q_{+}^{\Delta}(\lambda)}{\sin^{2}\frac{\lambda}{2}}\in L_{1}([0,\pi])\quad \mbox{and }\
  \phi^{-}(\lambda)=\frac{1}{\sin\frac{\lambda}{2}}\Bigl(\frac{Q_{-}^{\Delta}(\lambda)}{\sin\frac{\lambda}{2}}-A\Bigr)\in L_{1}([0,\pi]).
\]
Using this observation, we can rewrite the latter formal formula for $q_{k}$
in terms of $\phi^{+}(\lambda),\phi^{-}(\lambda)$: 
\begin{equation}
q_{n}=\frac{1}{4\pi}\int_{0}^{\pi}\phi^{+}(\lambda)\cos(n\lambda)d\lambda+\frac{1}{4\pi}\int_{0}^{\pi}\phi^{-}(\lambda)\sin(n\lambda)d\lambda+\frac{A}{4\pi}\int_{0}^{\pi}\frac{\sin(n\lambda)}{\sin(\lambda / 2)}\, d\lambda.\label{qkregdecompose}
\end{equation}
The last integral in this formula is known (see \cite{GR}, p.\thinspace 605, 3.612
(4)): 
\[
  \int_{0}^{\pi}\frac{\sin(n\lambda)}{\sin\lambda / 2}\, d\lambda=2\int_{0}^{\pi / 2}\frac{\sin(2nx)}{\sin x}dx=4\sum_{k=0}^{n-1}
  \frac{(-1)^{k}}{2k+1}=\pi+\bar{\bar{o}}(1),\ \mbox{as}\ n\rightarrow\infty.
\]
That is why the right-hand side of  (\ref{qkregdecompose})
defines a sequence in $l_{\infty}$. Denote it by $\widetilde{q}$:
\begin{equation}
\widetilde{q}_{n}=\frac{1}{4\pi}\int_{0}^{\pi}\phi^{+}(\lambda)\cos(n\lambda)d\lambda+\frac{1}{4\pi}\int_{0}^{\pi}\phi^{-}(\lambda)\sin(n\lambda)d\lambda+\frac{A}{4\pi}\int_{0}^{\pi}\frac{\sin(n\lambda)}{\sin(\lambda / 2)}\, d\lambda.\label{tildeqndef}
\end{equation}
It is easy to see that $\Delta\widetilde{q}=\Delta q$. Since $\widetilde{q},q\in l_{\infty}$
and $\Delta(\widetilde{q}-q)=0$,  there exists a constant $c$
such that 
\begin{equation}
q_{n}=c+\widetilde{q}_{n}\label{qntildeqnrel}
\end{equation}
for any $n\in\mathbb{Z}$. We want now to find the limit of $\widetilde{q}_{n}$
for large $|n|$. Since $\phi^{+}(\lambda),\phi^{-}(\lambda)\in L_{1}([0,\pi])$,
 by the Riemann\tire Lebesgue theorem the integrals in (\ref{tildeqndef})
containing these functions tend to zero as $|n|$ grows. Thus, 
\[
\lim_{n\rightarrow+\infty}\widetilde{q}_{n}=\frac{A}{4},\quad\lim_{n\rightarrow-\infty}\widetilde{q}_{n}=-\frac{A}{4}.
\]
From these equalities and formula (\ref{qntildeqnrel}) the theorem
follows. Formula (\ref{lplusminulimitplus}) will be proved during
the proof of Theorem \ref{limitTheorem}.

\section{Proof of Theorem \ref{sufficientCondld}}

Note first that condition (\ref{qvCondforlv}) implies that $Q^{\Delta}(\lambda)$
is $C^{1}$-smooth function, and the following equality holds: 
\begin{equation}
Q^{\Delta}(0)=0.\label{qdeltazeroeq}
\end{equation}
Smoothness follows from classical theorems for Fourier series. Let
us prove now (\ref{qdeltazeroeq}). By definition we have: 
\begin{equation}
Q^{\Delta}(0)=\sum_{k}q_{k}^{\Delta}.\label{qdeltazero}
\end{equation}
Moreover, the latter series is convergent. Write $q_{k}^{\Delta}$
as 
\[
q_{k}^{\Delta}=-(\delta_{k+1}-\delta_{k}),\ \quad\delta_{k}=q_{k}-q_{k-1}.
\]
Then from equality (\ref{qdeltazero}) it follows that 
\[
Q^{\Delta}(0)=\lim_{N\rightarrow+\infty}\sum_{|k|\leqslant N}q_{k}^{\Delta}=-\lim_{N\rightarrow+\infty}(\delta_{N}-\delta_{-N}).
\]
Let us show that $\delta_{N}\rightarrow0$ as $N\rightarrow\pm\infty$.
For this write $Q^{\Delta}(0)$ as follows: 
\[
Q^{\Delta}(0)=\sum_{k\geqslant0}q_{k}^{\Delta}+\sum_{k<0}q_{k}^{\Delta}.
\]
From absolute convergence in (\ref{qdeltazero}) it follows that in
the last formula both series converge. It is not difficult to see
that $\delta_{n}-\delta_{n+k},\ n,k\geqslant0,$ is the difference
of partial sums of the first series. Thus, by Cauchy principle, there
exists a finite limit 
\[
c=\lim_{n\rightarrow+\infty}\delta_{n}.
\]
Since for all $n,N\geqslant0$ we have 
\[
q_{N+n}=q_{N}+\sum_{k=N+1}^{n+N}\delta_{k}
\]
and $\sup_{n}|q_{n}|<\infty$, it follows that $c=0$. Similarly one can consider
the case $\lim_{n\rightarrow-\infty}\delta_{n}$. The equality (\ref{qdeltazeroeq})
is proved.

For the derivative of $Q^{\Delta}(\lambda)$ at zero we have: 
\[
\frac{d}{d\lambda}Q^{\Delta}(0)=i\sum_{k}kq_{k}^{\Delta}.
\]
In the formula (\ref{phiDef}), where the function $\phi(\lambda)$
was defined, put 
\[
A=2\frac{1}{i}\frac{d}{d\lambda}Q^{\Delta}(0)=2\sum_{k}kq_{k}^{\Delta}.
\]
Then we get equalities: 
\begin{align*}
  \phi(\lambda)&=\frac{1}{\sin(\lambda / 2)}\left(\frac{Q^{\Delta}(\lambda)}{\sin(\lambda / 2)}-iA\right)=
                 \frac{1}{\sin^{2} (\lambda / 2) }\sum_{k}q_{k}^{\Delta}\Bigl(e^{ik\lambda}-2ik\sin(\lambda / 2) \Bigr)\\
  &=\phi^{+}(\lambda)+i\phi^{-}(\lambda),
\end{align*}
where 
\[
\phi^{+}(\lambda)=\frac{1}{\sin^{2}(\lambda / 2)}\sum_{k}q_{k}^{\Delta}\cos(k\lambda),
\]
\[
\phi^{-}(\lambda)=\frac{1}{\sin^{2} (\lambda / 2)} \sum_{k}q_{k}^{\Delta}\Bigl(\sin(k\lambda)-2k\sin\frac{\lambda}{2}\Bigr).
\]
Since 
\[
|\phi^{-}(\lambda)|\leqslant\sum_{k}|q_{k}^{\Delta}|\left|\frac{\sin(k\lambda)-2k\sin\frac{\lambda}{2}}{\sin^{2}\frac{\lambda}{2}}\right|,
\]
from lemma \ref{vnineqlemma}, condition (\ref{qvCondforlv})
and Fatou theorem it follows that $\phi^{-}(\lambda)\in L_{1}([0,\pi])$.

Moreover, as $Q^{\Delta}(0)=0$, we have
\[
q_{0}^{\Delta}=-\sum_{k\ne0}q_{k}^{\Delta}.
\]
It follows that 
\[
\phi^{+}(\lambda)=\frac{1}{\sin^{2}(\lambda / 2)}\sum_{k}q_{k}^{\Delta}(\cos(k\lambda)-1).
\]
For all $x\in[0,\pi / 2]$ we have $\sin x\geqslant  2x / \pi$,
hence 
\begin{align*}
\int_{0}^{\pi}\frac{1-\cos k\lambda}{\sin^{2}(\lambda / 2)}\, d\lambda&=2\int_{0}^{\pi / 2}\frac{1-\cos2kx}{\sin^{2}x}dx\\
&\leqslant\pi\int_{0}^{\pi / 2} \frac{1-\cos2kx}{x^{2}}\, dx\leqslant\pi\int_{0}^{+\infty}\frac{1-\cos2kx}{x^{2}}dx=\pi^{2}|k|.
\end{align*}
We took an explicit formula for the latter integral  from \cite{GR},
p.\thinspace 691, 3.782 (2). Finally, from condition (\ref{qvCondforlv}) and
the Fatou theorem we get that $\phi^{+}(\lambda)\in L_{1}([0,\pi])$.
Theorem \ref{sufficientCondld} is proven.

\begin{lemma} \label{vnineqlemma} There exists a constant $c>0$ such
that for all $n>0$ we have the inequality 
\[
V_{n}=\int_{0}^{\pi}\Big| \frac{\sin(n\lambda)-2n\sin(\lambda / 2)}{\sin^{2} (\lambda / 2)}\Big| d\lambda\leqslant cn\ln n.
\]
\end{lemma}
\begin{proof}
  For any $x\in[0, \pi / 2]$ we have $\sin x\geqslant 2x / \pi $,
consequently, 
\begin{align}
    V_{n}&=2\int_{0}^{\pi / 2}\Bigl| \frac{\sin(2nx)-2n\sin x}{\sin^{2}x}\Bigr| dx \nonumber \\
    &\leqslant\frac{\pi^{2}}{2}\int_{0}^{\pi / 2}
         \Bigl| \frac{\sin(2nx)-2n\sin x}{x^{2}}\Bigr|dx   \label{Veintegie} \\
       &=\frac{\pi^{2}}{2}\biggl(\int_{0}^{\pi /(2n)}\Bigl| \frac{\sin(2nx)-2n\sin x}{x^{2}}\Bigr| dx+
       \int_{\pi / (2n)}^{\pi / 2}\Bigl| \frac{\sin(2nx)-2n\sin x}{x^{2}}\Bigr| dx\biggr). \nonumber 
\end{align}
Denote $h(x)=\sin(2nx)-2n\sin x$. The following equalities hold:
\[
h(0)=0,\ h'(0)=0,\quad h''(x)=-4n^{2}\sin(2nx)+2n\sin x.
\]
Then for all $x$ we have the inequality: 
\[
|h(x)|\leqslant3n^{2}x^{2}.
\]
Using this inequality, we can get the following bound for the integral
in (\ref{Veintegie}): 
\[
\int_{0}^{\pi /(2n)}\Bigl| \frac{\sin(2nx)-2n\sin x}{x^{2}}\Bigr| dx\leqslant3n^{2}\int_{0}^{\pi / (2n)}dx=\frac{3\pi}{2}n.
\]
To estimate the second integral note that $h(x)$ satisfies the following
inequality: 
\[
|h(x)|\leqslant4nx.
\]
Consequently, the second term in (\ref{Veintegie}) admits the 
bound: 
\[
  \int\limits_{\pi / (2n)}^{\pi /2} \Bigl| \frac{\sin(2nx)-2n\sin x}{x^{2}}\Bigr| dx \leqslant 4n\int\limits_{\pi / (2n)}^{\pi / 2} \frac{1}{x}dx=
  4n\Bigl(\ln\frac{\pi}{2}-\ln\frac{\pi}{2n}\Bigr)=4n\ln n.
\]
Thereby Lemma \ref{vnineqlemma} is proved.
\end{proof}

\section{Proof of Theorem \ref{unBoundTh} }

\begin{lemma} For any $n\in\mathbb{Z}$ and $t\geqslant0$ the following
equality holds true: 
\begin{equation}
  q_{n}(t)=q_{n}(0)-\frac{1}{2\pi}\int_{0}^{2\pi}\frac{1-\cos(2\omega t\sin (\lambda / 2))}{4\sin^{2}(\lambda / 2)}
  Q^{\Delta}(\lambda)e^{-in\lambda}d\lambda.\label{solqviaqvfol}
\end{equation}
\end{lemma}
\begin{proof}
Denote by $R_{n}(t)$ the right-hand side in (\ref{solqviaqvfol}).
Since $Q^{\Delta}\in L_{2}([0,2\pi])$,  we get the
inequality for $R_{n}(t)$: 
\begin{align*}
  |R_{n}(t)| &\leqslant|q_{n}(0)|+\frac{1}{2\pi}\sqrt{\int_{0}^{2\pi}|Q^{\Delta}(\lambda)|^{2}d\lambda}
               \sqrt{\int_{0}^{2\pi}\frac{(1-\cos(2\omega t\sin(\lambda / 2)))^{2}}{16\sin^{4} (\lambda / 2)}d\lambda}\\
  &\leqslant|q_{n}(0)|+t^{2}||Q^{\Delta}||_{L_{2}([0,2\pi])}.
\end{align*}
Thereby, $R_{n}(t)$ for any $t\geqslant0$ defines the sequence in
$l_{\infty}$. Now differentiating $R_{n}(t)$ two times, we get 
\[
\ddot{R}_{n}(t)=-\omega^{2}\frac{1}{2\pi}\int_{0}^{2\pi}\cos\Bigl( 2\omega t\sin\frac{\lambda}{2}\Bigr) Q^{\Delta}(\lambda)e^{-in\lambda}d\lambda .
\]
From the other side, 
\begin{align*}
\Delta R(t)&=R_{n+1}+R_{n-1}-2R_{n}
\\
           &=(\Delta q(0))_{n}-\frac{1}{2\pi}\int_{0}^{2\pi}\frac{1-\cos(2\omega t\sin(\lambda / 2))}{4\sin^{2} (\lambda / 2)} \, Q^{\Delta}(\lambda)\\
&\quad {} \times             (e^{-i(n+1)\lambda}+e^{-i(n-1)\lambda}-2e^{-in\lambda})d\lambda
\\
           &=-q_{n}^{\Delta}(0)-\frac{1}{2\pi}\int_{0}^{2\pi}\frac{1-\cos(2\omega t\sin(\lambda / 2))}{4\sin^{2} (\lambda / 2)}Q^{\Delta}(\lambda)
             e^{-in\lambda}(2\cos\lambda-2)d\lambda
\\
&=-q_{n}^{\Delta}(0)+\frac{1}{2\pi}\int_{0}^{2\pi}\Bigl( 1-\cos\Bigl(2\omega t\sin\frac{\lambda}{2}\Bigr)\Bigr) Q^{\Delta}(\lambda)e^{-in\lambda}d\lambda
\\
           &=-q_{n}^{\Delta}(0)+\frac{1}{2\pi}\int_{0}^{2\pi}Q^{\Delta}(\lambda)e^{-in\lambda}d\lambda\\
  &\quad {} -\frac{1}{2\pi}\int_{0}^{2\pi}
             \cos\Bigl(2\omega t\sin\frac{\lambda}{2}\Bigr) Q^{\Delta}(\lambda)e^{-in\lambda}d\lambda
\\
&=-\frac{1}{2\pi}\int_{0}^{2\pi}\cos\Bigl( 2\omega t\sin\frac{\lambda}{2}\Bigr) Q^{\Delta}(\lambda)e^{-in\lambda}d\lambda.
\end{align*}
We can conclude now that $\ddot{R}(t)\! = \! \omega^{2}\Delta R(t)$ and
$R(0)=q(0)$. Due to the uniqueness of solution of the corresponding
ODE in $l_{\infty}$, the lemma is proved.
\end{proof}

The formula for solution in terms of $Q^{\Delta}(\lambda)$ can be
obtained ``directly''. It is known that 
\[
q(t)=\cos(t\sqrt{V})q(0),
\]
\[
q_{n}(t)=\sum_{k}a_{k}(t)q_{n-k}(0)
\]
(see \cite{LM_1}, lemma 3.1 and 3.3). 
This formula defines the action of the operator $\cos(t\sqrt{V})$.

Furthermore, from 
\[
\ddot{q}=-\cos(t\sqrt{V})Vq(0)=-\omega^{2}\cos(t\sqrt{V})q^{\Delta}(0)
\]
we get that 
\[
\ddot{q}_{n}(t)=-\omega^{2}\sum_{k}a_{k}(t)q_{n-k}^{\Delta}(0).
\]
After integrating with respect to $t$ and some transformations we can get (\ref{solqviaqvfol}).

Now we can come back to the proof of Theorem \ref{unBoundTh} concerning
uniform boundedness. We use the representation (\ref{QDeltaSymAsymF})
for $Q^{\Delta}$ : 
\[
Q^{\Delta}(\lambda)=Q_{+}^{\Delta}(\lambda)+iQ_{-}^{\Delta}(\lambda).
\]
Using formula (\ref{solqviaqvfol}) and that $q_{n}(t),Q_{+}^{\Delta},Q_{-}^{\Delta}$
are real, we get 
\begin{equation}
q_{n}(t)=q_{n}(0)-D_{n}(t)-B_{n}(t),\label{qndbnfor}
\end{equation}
where 
\begin{align*}
  D_{n}(t)&=\frac{1}{2\pi}\int_{0}^{2\pi}\frac{1-\cos(2\omega t\sin (\lambda / 2))}{4\sin^{2}(\lambda / 2)} \, Q_{+}^{\Delta}(\lambda)\cos(n\lambda)d\lambda
  \\
  &=\frac{1}{\pi}\int_{0}^{\pi}\frac{1-\cos(2\omega t\sin(\lambda / 2))}{4\sin^{2} (\lambda / 2)}\, Q_{+}^{\Delta}(\lambda)\cos(n\lambda)d\lambda ,
\\
  B_{n}(t)&=\frac{1}{2\pi}\int_{0}^{2\pi}\frac{1-\cos(2\omega t\sin(\lambda / 2))}{4\sin^{2}(\lambda / 2)}\, Q_{-}^{\Delta}(\lambda)\sin(n\lambda)d\lambda \\
  &=\frac{1}{\pi}\int_{0}^{\pi}\frac{1-\cos(2\omega t\sin(\lambda / 2))}{4\sin^{2}(\lambda / 2)}\, Q_{-}^{\Delta}(\lambda)\sin(n\lambda)d\lambda .
\end{align*}
In the second equality for $D_{n},B_{n}$ we used the symmetry of
integrands with respect to point $\pi$.

It is not difficult to see (from the second condition of (\ref{phiDef})
in the definition of $l^{\Delta}$) that  $\phi\in L_{1}([0,\pi])$
iff 
\[
  \phi^{+}(\lambda)=\frac{Q_{+}^{\Delta}(\lambda)}{\sin^{2}(\lambda / 2)}\in L_{1}([0,\pi])
\]
and
\[
  \phi^{-}(\lambda)=\frac{1}{\sin(\lambda / 2)}\Bigl(\frac{Q_{-}^{\Delta}(\lambda)}{\sin(\lambda / 2)}-A\Bigr)\in L_{1}([0,\pi]) .
\]
Thus, for $D_{n}(t)$ we have 
\[
|D_{n}(t)|\leqslant\frac{1}{4\pi}\int_{0}^{\pi}\frac{|Q_{+}^{\Delta}(\lambda)|}{\sin^{2} (\lambda / 2)}d\lambda.
\]
This means that $D_{n}(t)$ is uniformly bounded. Also, for $B_{n}(t)$
we have 
\begin{align*}
  B_{n}(t)&=\frac{1}{\pi}\int_{0}^{\pi}\frac{1-\cos(2\omega t\sin (\lambda / 2))}{4\sin(\lambda / 2)}
            \Bigl(\frac{Q_{-}^{\Delta}(\lambda)}{\sin (\lambda / 2)}-A\Bigr)\sin(n\lambda)d\lambda
\\
&\quad {} +\frac{A}{\pi}\int_{0}^{\pi}\frac{1-\cos(2\omega t\sin (\lambda / 2))}{4\sin (\lambda / 2)}\sin(n\lambda)d\lambda .
\end{align*}
Then by Lemma \ref{mainGestIneq} we get the bound: 
\[
|B_{n}(t)|\leqslant\frac{1}{4\pi}\int_{0}^{\pi}|\phi^{-}(\lambda)|d\lambda+\frac{AC}{4\pi}
\]
for some constant $C>0$. Theorem \ref{unBoundTh} follows.

\begin{lemma} \label{mainGestIneq} There exists a constant $C$ such
that for all $t\geqslant0$ and all $n\in\mathbb{Z}$ the following
inequality holds: 
\[
\biggl|\int_{0}^{\pi}\frac{1-\cos(t\sin (\lambda / 2))}{\sin (\lambda / 2)}\sin(n\lambda)d\lambda\biggr|\leqslant C.
\]
\end{lemma}

  Without loss of generality we can assume that $n\geqslant0$.
Then we have 
\begin{equation}
  \int_{0}^{\pi}\frac{1-\cos(t\sin (\lambda / 2))}{\sin (\lambda / 2)}\sin(n\lambda)d\lambda=
  \int_{0}^{\pi}\frac{\sin(n\lambda)}{\sin (\lambda / 2)}d\lambda-2I_{2n}(t),\label{201905121}
\end{equation}
where 
\[
  I_{n}(t)=\frac{1}{2}\int_{0}^{\pi}\frac{\cos(t\sin (\lambda / 2))}{\sin (\lambda / 2)}
  \sin\Bigl(n\frac{\lambda}{2}\Bigr)d\lambda=\int_{0}^{\pi / 2} \frac{\cos(t\sin x)}{\sin x}\sin(nx)dx.
\]
The integral in (\ref{201905121})  can be found  in 
\cite{GR}, p.\thinspace 605, 3.612 (4): 
\[
  \int_{0}^{\pi}\frac{\sin(n\lambda)}{\sin (\lambda / 2)}d\lambda=2\int_{0}^{\pi / 2} \frac{\sin(2nx)}{\sin x}dx
  =4\sum_{k=0}^{n-1}\frac{(-1)^{k}}{2k+1}=\pi+\bar{\bar{o}}(1),\quad \mbox{as}\ n\rightarrow\infty.
\]
Thus, the statement of  Lemma \ref{mainGestIneq} is equivalent
to the uniform boundedness of $I_{n}(t)$. \begin{theorem} \label{Intlemma}
The function 
\[
I_{n}(t)=\int_{0}^{\pi / 2}\frac{\cos(t\sin x)\sin(nx)}{\sin x}dx
\]
is bounded uniformly in $n\in\mathbb{Z}$ and $t\geqslant0$.
\end{theorem}

\subsection{Uniform boundedness and integrals of Bessel functions}

We could find only rather cumbersome proof of Theorem \ref{Intlemma}.
Before starting the proof, we show the connection of the function $I_{n}(t)$
with Bessel function $J_{n}(t)$ of the first kind 
\[
J_{n}(t)=\frac{1}{\pi}\int_{0}^{\pi}\cos(nx-t\sin x)dx.
\]
We have the equality 
\begin{align*}
\frac{d}{dt}I_{n}(t)&=-\int_{0}^{\pi / 2}\sin(t\sin x)\sin(nx)dx
\\
                    &=\frac{1}{2}\int_{0}^{\pi / 2} \cos(nx+t\sin x)dx-\frac{1}{2}\int_{0}^{\pi / 2} \cos(nx-t\sin x)dx\\
  &=\frac{1+(-1)^{n}}{2}\int_{0}^{\pi / 2} \cos(nx+t\sin x)dx-\frac{\pi}{2}J_{n}(t).
\end{align*}
Introducing the notation $\chi_{n}=[1+(-1)^{n}] / 2$, we get
\begin{align*}
I_{n}(t)&=I_{n}(0)+\int_{0}^{t}\frac{d}{ds}I_{n}(s)ds
\\
        &=I_{n}(0)+\chi_{n}\int_{0}^{\pi / 2}\frac{\sin(nx+t\sin x)}{\sin x}dx\\
  &\quad {} -\chi_{n}\int_{0}^{\pi / 2}\frac{\sin(nx)}{\sin x}dx-\frac{\pi}{2}\int_{0}^{t}J_{n}(s)ds
\\
&=(1-\chi_{n})I_{n}(0)+\chi_{n}\int_{0}^{\pi / 2}\frac{\sin(2nx+t\sin x)}{\sin x}dx+\frac{\pi}{2}\int_{0}^{t}J_{n}(s)ds.
\end{align*}
Boundedness of $I_{n}(0)$ follows from \cite{GR}, p.\thinspace 605, 3.612
(3)--(4). We are going to show now the uniform boundedness of the second
integral in the last formula 
\begin{align*}
\int_{0}^{\pi / 2}\frac{\sin(2nx+t\sin x)}{\sin x}dx&=\int_{0}^{\pi / 2}\frac{\sin(2nx+t\sin x)}{x}\frac{x}{\sin x}dx
\\
                                                    &=\int_{0}^{\pi / 2}\frac{\sin(2nx+t\sin x)}{x}\Bigl(\frac{x}{\sin x}-1\Bigr)dx\\
  &\quad {} +\int_{0}^{\pi / 2} \frac{\sin(2nx+t\sin x)}{x}dx.
\end{align*}
Since the function $\frac{1}{x}\left(\frac{x}{\sin x}-1\right)$ is absolutely
integrable on $[0,\pi/2]$, the first integral is uniformly bounded.
Similarly, for the second integral we can use Lemma \ref{rntemestimate}.
Thus we have shown that uniform boundedness of $I_{n}(t)$ is equivalent
to the uniform boundedness in $n\in\mathbb{Z}_{+}$ and $t\geqslant0$
of the integral of the Bessel function 
\begin{equation}
G_{n}(t)=\int_{0}^{t}J_{n}(s)ds.\label{G_n_t}
\end{equation}

Unfortunately, we could not find any results concerning uniform (in
index and time) boundedness of the integral of Bessel function (\ref{G_n_t}).
For example, in \cite{Watson} (see p.\thinspace 255, formula (11)) 
only the asymptotics of $G_{n}(nx)$ for large $n$ and $0<x\leqslant1$ is given.
This is not sufficient for uniform boundedness in $t,n$. One of the
reasons why it is difficult to get uniform estimates for $G_{n}(t)$
is the following. There exist many asymptotic formulas for $J_{n}(t)$
with different speed of $n$ and $t$ increase to infinity, but we
could not find any for $J_{n}(nx)$ uniformly for $x\in[1-\delta,1+\delta]$
with $\delta>0$. As a corollary of Theorem \ref{Intlemma} we get
the following statement.

\begin{proposition} There exists a constant $C>0$ such that for all
integers $n$ and all $t\geqslant0$ the following inequality holds
\[
\Biggl| \int_{0}^{t}J_{n}(s)ds\Biggr| \leqslant C.
\]
\end{proposition}

We want to show now the relation of the integral
$G_{n}(t)$ with the sum of $J_{n}(t)$. For simplicity assume that $n=2m$
for some integer $m>0$. Using the known formula (see \cite{Watson})
\[
2\dot{J}_{n}=J_{n-1}+J_{n+1}
\]
we conclude that 
\begin{align*}
  G_{2m}(t)&=\int_{0}^{t}J_{2m}(s)ds=2J_{2m-1}(t)-\int_{0}^{t}J_{2m-2}(s)ds=\ldots\\
  &=2\sum_{k=0}^{m-1}(-1)^{m-1-k}J_{2k+1}(t)-\int_{0}^{t}J_{0}(s)ds.
\end{align*}
For all $k\geqslant0$ we have 
\[
\int_{0}^{+\infty}J_{k}(t)dt=1,
\]
(see \cite{GR}, p.\thinspace 1036, 6.511 (1)), then uniform boundedness of
$G_{2m}(t)$ is equivalent to the uniform boundedness (in $m$ and
$t$) of the sums 
\[
\sum_{k=0}^{m-1}(-1)^{k}J_{2k+1}(t).
\]
(Unfortunately, we could not find anything concerning this in
the literature.) And again, as a corollary to Theorem \ref{Intlemma}
we have the following statement.

\begin{proposition} There exists a constant $C>0$ such that for
any integer $n\geqslant0$ and all $t\geqslant0$ the following inequalities
hold 
\[
\biggl| \sum_{k=0}^{n-1}(-1)^{k}J_{2k+1}(t)\biggr| \leqslant C,\quad\biggl| \sum_{k=1}^{n}(-1)^{k}J_{2k}(t)\biggr| \leqslant C.
\]
\end{proposition}

Also remark that in a very informative book \cite{Luke}, where integrals of Bessel functions are considered (for example,
see pp.\thinspace 58--60)
we also could not find uniform estimates for $G_{n}(t)$.

\subsection{Proof of Theorem \ref{limitTheorem} }

We will use the representation (\ref{qndbnfor}) of the solution as
\[
q_{n}(t)=q_{n}(0)-D_{n}(t)-B_{n}(t).
\]
Using the Riemann\tire Lebesgue theorem we conclude that 
\[
\lim_{t\rightarrow\infty}D_{n}(t)=\frac{1}{4\pi}\int_{0}^{\pi}\phi^{+}(\lambda)\cos(n\lambda)d\lambda,
\]
\[
\lim_{t\rightarrow\infty}B_{n}(t)=\frac{1}{4\pi}\int_{0}^{\pi}\phi^{-}(\lambda)\sin(n\lambda)d\lambda+\frac{A}{4\pi}\int_{0}^{\pi}\frac{\sin(n\lambda)}{\sin (\lambda / 2)}d\lambda.
\]
It follows that 
\[
\lim_{t\rightarrow\infty}q_{n}(t)=q_{n}(0)-\widetilde{q}(0)_{n},
\]
where the sequence $\widetilde{q}(0)$ is defined in (\ref{tildeqndef}).
Therein it was shown that there exists a constant $c\in\mathbb{R}$
such that for all $n\in\mathbb{Z}$ 
\[
q_{n}(0)-\widetilde{q}(0)_{n}=c.
\]
In the limits  $n\rightarrow+\infty$ and $n\rightarrow-\infty$,
we get: 
\[
L_{+}-\frac{A}{4}=c,\quad L_{-}+\frac{A}{4}=c.
\]
It follows that $(L_{+}+L_{-})/2=c$. This proves Theorem \ref{limitTheorem}
and formula (\ref{lplusminulimitplus}).

\subsection{Proof of Theorem \ref{Intlemma}}

Firstly, transform our integral as follows: 
\begin{align*}
I_{n}(t)&=\int_{0}^{\pi / 2}\frac{\cos(t\sin x)\sin(nx)}{\sin x}dx=\int_{0}^{\pi / 2}\frac{\cos(t\sin x)\sin(nx)}{x}\frac{x}{\sin{x}}dx
\\
&=\int_{0}^{\pi / 2}\cos(t\sin x)\sin(nx)\frac{1}{x}\Bigl(\frac{x}{\sin{x}}-1\Bigr)dx+\int_{0}^{\pi / 2}\frac{\cos(t\sin x)\sin(nx)}{x}dx.
\end{align*}
Since the function $\frac{1}{x}\left(\frac{x}{\sin{x}}-1\right)$ is
absolutely integrable on $[0,\pi/2]$, the uniform boundedness
of $I_{n}(t)$ is equivalent to the uniform boundedness of the integral
\[
C_{n}(t)=\int_{0}^{\pi / 2}\frac{\cos(t\sin x)\sin(nx)}{x}dx.
\]
Further on we shall prove uniform boundedness of $C_{n}(t)$.

Now it is useful to say some words about the scheme of the proof.
The proof will be subdivided in 4 parts depending on the parameters
$n$ and $t$: 
\begin{enumerate}
\item $\gamma_{n}(t)= t / n \leqslant\gamma_{1}<1$, lemma \ref{lessonelemma}.
In this domain we use integration by parts several times and will
use Gronwall's lemma. 
\item $\gamma_{n}(t)\geqslant\gamma_{2}>1$, lemma \ref{greaterOnelemma}.
The proof is similar to the previous case. 
\item $\gamma_{1}\leqslant\gamma_{n}(t)\leqslant1$, lemma \ref{gammalessonelemma}.
Here we use ``regularized'' change of variables to reduce the problem
to tabulated integral. 
\item $1\leqslant\gamma_{n}(t)\leqslant\gamma_{2}$, lemma \ref{gammagreateronelemma}.
Here we use another ``regularized'' change of variables, enter complex
plane and estimate the obtained integrals. 
\end{enumerate}
We will choose  the constants $\gamma_{1}<1<\gamma_{2}$ appropriately
during the proof. Let us explain the reasons for such partition onto
4 parts of the set of parameters $n$ and $t$. Firstly note that
the integral in question is an oscillation integral having singularity
at the ends of integration interval. For small times $t \ll n$ the main
contribution is given by the terms $\sin(nx)$. Contrary to this, for large
times $t \gg n$ the main contribution is expected from the terms $\cos(t\sin x)$.
Moreover in the latter case the phase function $\sin x$ has stationary
point $\pi / 2$. In both cases we could deal with oscillations
using  integration by parts several times. In case $t \gg n$ it is also necessary
to shift the integration region from stationary points. When $t\sim n$
both terms give oscillations of the same order. That is why in
this case we ``put'' these terms under the same $\sin$. That gives
two new integrals which should be considered separately.

Now we start the detailed proof.

\subsection{Uniform boundedness for $\boldsymbol{\gamma_{n}(t)\leqslant\gamma_{1}<1}$}

\begin{lemma}[Domain $\boldsymbol{\gamma_{n}(t)\leqslant\gamma_{1}<1}$]
\label{lessonelemma} There exists a constant $c>0$ such that for all
$0<\gamma_{1}<1$, all $n>0$ and $t\geqslant0$ satisfying the condition
$\gamma_{n}(t)= t/n \leqslant\gamma_{1}$, the following inequality
holds 
\[
|C_{n}(t)|\leqslant\frac{c}{1-\gamma_{1}^{2}}\exp\Bigl(\frac{1}{1-\gamma_{1}^{2}}\Bigr).
\]
\end{lemma}
\begin{proof} We have 
\begin{align*}
  C_{n}(t)&=\int_{0}^{\pi / 2}\frac{(\cos(t\sin x)-1)\sin(nx)}{x}dx+\int_{0}^{\pi / 2}\frac{\sin(nx)}{x}dx\\
  &=\int_{0}^{\pi / 2} f_{t}(x)\sin(nx)dx+O(1)
\end{align*}
where 
\[
f_{t}(x)=\frac{\cos(t\sin x)-1}{x}.
\]
Here and further we shall write $f(t,n,\gamma_1) = O(g(t,n,\gamma_1))$ if for all $t,n,\gamma_1$ such that $t/n\leqslant \gamma_1 \leqslant 1$ the following inequality holds
$$
|f(t,n,\gamma_1)| \leqslant w g(t,n,\gamma_1)
$$
for some constant $w$ not depending on $t,n,\gamma_1$. In other words, it is a well known notation, 
but with additional condition that the corresponding constant should not depend on our parameters.
Integrating by parts we get 
\begin{align*}
  C_{n}(t)&=-\frac{f_{t}(x)\cos(nx)}{n}\Big|_{0}^{\pi / 2}+\frac{1}{n}\int_{0}^{\pi / 2}f'_{t}(x)\cos(nx)dx\\
  &=\frac{1}{n} f_{t}\Bigl(\frac{\pi}{2}\Bigr)\cos\Bigl(\frac{n\pi}{2}\Bigr) +\frac{1}{n}\int_{0}^{\pi / 2}f'_{t}(x)\cos(nx)dx.
\end{align*}
In the latter equality we used that $f_{t}(0)=0$. Further on, we
have 
\[
f_{t}'(x)=-\frac{t\cos(x)\sin(t\sin x)}{x}-\frac{\cos(t\sin x)-1}{x^{2}}.
\]
Hence
\begin{align}
  C_{n}(t)&=O(1)-\frac{t}{n}\int_{0}^{\pi / 2}\frac{\cos(x)\sin(t\sin x)}{x}\cos(nx)dx\nonumber \\
  &\quad {} -\frac{1}{n}\int_{0}^{\pi / 2}\frac{\cos(t\sin x)-1}{x^{2}}\cos(nx)dx.\label{CntPartOne}
\end{align}
Let us show first that the last integral in (\ref{CntPartOne}) is
of the order $O(t)$. We have the inequalities 
\begin{align*}
  &\Biggl|\int_{0}^{\pi / 2}\frac{\cos(t\sin x)-1}{x^{2}}\cos(nx)dx\Biggr|\\
  &\quad 
               \leqslant\int_{0}^{\pi / 2}\frac{1-\cos(t\sin x)}{x^{2}}dx=2\int_{0}^{\pi / 2}\frac{\sin^{2}((t\sin x) /2)}{x^{2}}dx
\\
                                             &\quad =2\int_{0}^{\pi / 4}\frac{\sin^{2}((t\sin x)/2)}{x^{2}}dx+2\int_{\pi / 4}^{\pi / 2}\frac{\sin^{2}((t\sin x)/2)}{x^{2}}dx\\
  &\quad =2\int_{0}^{\pi / 4}\frac{\sin^{2}((t\sin x)/2)}{x^{2}}dx+O(1)
\\
                                                                         &\quad \leqslant2\int_{0}^{\pi / 4}\frac{\sin^{2}((t\sin x) / 2)}{\sin^{2}x}dx+O(1).
\end{align*}
                    Changing the variables           $u=\sin x$, we have
                    \begin{align*}
                                                                        2\int_{0}^{1 / \sqrt{2}}\frac{\sin^{2}(tu / 2)}{u^{2}\sqrt{1-u^{2}}}du+O(1)&\leqslant2\sqrt{2}
                                                                          \int_{0}^{1 / \sqrt{2}}\frac{\sin^{2}(tu/2)}{u^{2}}du+O(1)\\
  &=\sqrt{2}t\int_{0}^{t / (2\sqrt{2})}\frac{\sin^{2}(y)}{y^{2}}dy+O(1)
\\
                                                                              &=\sqrt{2}t(O(1)+\int_{1}^{t / (2\sqrt{2})}\frac{\sin^{2}(y)}{y^{2}}dy)+O(1)\\
  &\leqslant\sqrt{2}t(O(1)+\frac{1}{t}))+O(1)=O(t)+O(1).
\end{align*}
Now transform the first integral in (\ref{CntPartOne}) as follows:
\begin{align*}
\int_{0}^{\pi / 2}\frac{\cos(x)\sin(t\sin x)}{x}\cos(nx)dx&=\int_{0}^{\pi / 2}\frac{\cos(x)-1}{x}\sin(t\sin x)\cos(nx)dx\\
                                                          &\quad {} +\int_{0}^{\pi / 2}\frac{\sin(t\sin x)\cos(nx)}{x}dx\\
  &=O(1)+S_{n}(t),
\end{align*}
where 
\[
S_{n}(t)=\int_{0}^{\pi / 2}\frac{\sin(t\sin x)\cos(nx)}{x}dx.
\]
We used the fact that the function $[\cos(x)-1] / x$ is absolutely
integrable on $[0,\pi/2]$. From (\ref{CntPartOne}) we
have 
\begin{equation}
C_{n}(t)=O(1)+O\Bigl(\frac{t}{n}\Bigr)-\frac{t}{n}S_{n}(t)=O(1)-\frac{t}{n}S_{n}(t).\label{CntFirstPartSum}
\end{equation}
Now integrate $S_{n}(t)$ by parts: 
\begin{align*}
S_{n}(t)&=\frac{\sin(t\sin x)\sin(nx)}{nx}\Big|_{0}^{\pi / 2}-\frac{1}{n}\int_{0}^{\pi / 2}\Bigl(\frac{\sin(t\sin x)}{x}\Bigr)'\sin(nx)dx\\
&=O(1)-\frac{1}{n}\int_{0}^{\pi / 2}\Bigl(\frac{t\cos(x)\cos(t\sin x)}{x}-\frac{\sin(t\sin x)}{x^{2}}\Bigr)\sin(nx)dx
\\
        &=O(1)-\frac{t}{n}\int_{0}^{\pi / 2}\frac{\cos(x)\cos(t\sin x)}{x}\sin(nx)dx\\
  &\quad {} +\frac{1}{n}\int_{0}^{\pi / 2}\frac{\sin(t\sin x)}{x^{2}}\sin(nx)dx
\\
&=-\frac{t}{n}H_{n}(t)+\frac{1}{n}F_{n}(t).
\end{align*}
For the first integral we have: 
\begin{align*}
  H_{n}(t)&=\int_{0}^{\pi / 2}\frac{\cos(x)\cos(t\sin x)}{x}\sin(nx)dx\\
  &=\int_{0}^{\pi / 2}\frac{(\cos(x)-1)}{x}\cos(t\sin x)\sin(nx)dx
\\
          &\quad {} +\int_{0}^{\pi / 2}\frac{\cos(t\sin x)\sin(nx)}{x}dx\\
  &=O(1)+C_{n}(t).
\end{align*}
Note that for the second integral 
\[
F_{n}(t)=\int_{0}^{\pi / 2}\frac{\sin(t\sin x)}{x^{2}}\sin(nx)dx
\]
the following equalities hold: 
\[
\frac{d}{dt}F_{n}(t)=\int_{0}^{\pi / 2}\frac{\sin(x)}{x}\frac{\cos(t\sin x)\sin(nx)}{x}dx=O(1)+C_{n}(t).
\]
It follows that 
\[
F_{n}(t)=O(t)+\int_{0}^{t}C_{n}(s)ds.
\]
Thus, we have got the equality 
\[
S_{n}(t)=O\Bigl(\frac{t}{n}\Bigr)-\frac{t}{n}C_{n}(t)+\frac{1}{n}\int_{0}^{t}C_{n}(s)ds=O(1)-\frac{t}{n}C_{n}(t)+\frac{1}{n}\int_{0}^{t}C_{n}(s)ds .
\]
Substitute it to (\ref{CntFirstPartSum}) and get 
\[
C_{n}(t)=O(1)+\Bigl(\frac{t}{n}\Bigr)^{2}C_{n}(t)-\frac{t}{n^{2}}\int_{0}^{t}C_{n}(t)ds.
\]
Thus for $\gamma_{n}(t)=t / n \ne1$ we get 
\[
C_{n}(t)=\frac{1}{1-(t / n)^{2}}\Biggl(O(1)-\frac{t}{n^{2}}\int_{0}^{t}C_{n}(t)ds\Biggr).
\]
We have that if $\gamma_{n}(t)<\gamma_{1}<1$, then for some constant
$c>0$,  
\[
|C_{n}(t)|\leqslant\frac{1}{1-\gamma_{1}^{2}}\Biggl(c+\frac{1}{n}\int_{0}^{t}|C_{n}(t)|ds\Biggr).
\]
From Gronwall's lemma we get that 
\[
  |C_{n}(t)|\leqslant\frac{c}{1-\gamma_{1}^{2}}\exp\Bigl(\frac{t}{n(1-\gamma_{1}^{2})}\Bigr)\leqslant\frac{c}{1-\gamma_{1}^{2}}
  \exp\Bigl(\frac{1}{1-\gamma_{1}^{2}}\Bigr).
\]
Lemma \ref{lessonelemma} is thus proved.
\end{proof}

\subsection{Uniform boundedness for $\boldsymbol{\gamma_{n}(t)\geqslant\gamma_{2}>1}$}

\begin{lemma}[Domain $\boldsymbol{\gamma_{n}(t)\geqslant\gamma_{2}>1}$]
\label{greaterOnelemma} There exists a constant $c>0$ such that for
all $\gamma_{2}>1$, all $n>0$ and $t>0$, satisfying condition $\gamma_{n}(t)=t/n\geqslant\gamma_{2}$,
the following inequality holds: 
\[
|C_{n}(t)|\leqslant c\frac{\gamma_{2}^{2}}{\gamma_{2}^{2}-1}\exp\Bigl(\frac{\gamma_{2}}{\gamma_{2}^{2}-1}\Bigr) .
\]
\end{lemma}
\begin{proof}
  Change first the integration domain 
\begin{align*}
C_{n}(t)&=\int_{0}^{\pi / 4}\frac{\cos(t\sin x)\sin(nx)}{x}dx+\int_{\pi / 4}^{\pi / 2}\frac{\cos(t\sin x)\sin(nx)}{x}dx
\\
&=\tilde{C}_{n}(t)+O(1),
\end{align*}
where 
\[
\tilde{C}_{n}(t)=\int_{0}^{\pi / 4}\frac{\cos(t\sin x)\sin(nx)}{x}dx.
\]
Integrating by parts we get 
\begin{align}
\tilde{C}_{n}(t)&=\frac{\sin(t\sin x)}{t\cos x}\frac{\sin(nx)}{x}\Big|_{0}^{\pi / 4}-\frac{1}{t}\int_{0}^{\pi / 4}\sin(t\sin x)\left(\frac{\sin(nx)}{x\cos x}\right)'dx
\nonumber \\
&=O(1)-\frac{1}{t}\int_{0}^{\pi / 4}\sin(t\sin x)\left(\frac{n\cos(nx)}{x\cos x}-\frac{\sin(nx)}{x^{2}\cos x}+\frac{\sin(nx)\sin x}{x\cos^{2}x}\right)dx \nonumber
\\
&=O(1)-\frac{n}{t}\int_{0}^{\pi / 4}\sin(t\sin x)\frac{\cos(nx)}{x\cos x}dx+\frac{1}{t}\int_{0}^{\pi / 4}\sin(t\sin x)\frac{\sin(nx)}{x^{2}\cos x}dx
\nonumber
\\
&=O(1)-\frac{n}{t}K_{n}(t)+\frac{1}{t}U_{n}(t) \label{C_tilde_n_t}
\end{align}
where 
\[
K_{n}(t)=\int_{0}^{\pi / 4}\sin(t\sin x)\frac{\cos(nx)}{x\cos x}dx,\quad U_{n}(t)=\int_{0}^{\pi / 4}\sin(t\sin x)\frac{\sin(nx)}{x^{2}\cos x}dx.
\]
In the equality (\ref{C_tilde_n_t}) we used the fact that 
\[
\Bigl|\sin(t\sin x)\frac{\sin(nx)\sin x}{x\cos^{2}x}\Bigr|=\Bigl|\sin(t\sin x)\frac{\sin(nx)}{\cos^{2}x}\frac{\sin x}{x}\Bigr|\leqslant 1
\]
for $x\in[0,\pi/4]$ and 
\[
\frac{1}{t}\leqslant\frac{1}{\gamma_{2}n}\leqslant1.
\]
Thus we got the expansion 
\begin{equation}
\tilde{C}_{n}(t)=O(1)-\frac{n}{t}K_{n}(t)+\frac{1}{t}U_{n}(t).\label{tildecnt}
\end{equation}
We will now estimate $K_{n}(t)$ and $U_{n}(t)$. We use equalities
\begin{align*}
U_{n}(t)&=\int_{0}^{\pi / 4}\sin(t\sin x)\frac{\sin(nx)}{x^{2}}\Bigl(\frac{1}{\cos x}-1\Bigr)dx+\int_{0}^{\pi / 4}\sin(t\sin x)\frac{\sin(nx)}{x^{2}}dx
\\
        &=O(1)+\int_{0}^{\pi / 4}\sin(t\sin x)\frac{\sin(nx)}{\sin^{2}x}\frac{\sin^{2}x}{x^{2}}dx\\
  &=O(1)+\int_{0}^{\pi / 4}\sin(t\sin x)\frac{\sin(nx)}{\sin^{2}x}dx
\\
&=O(1)+\int_{0}^{\pi / 4}\frac{\sin(t\sin x)}{\sin x}\frac{\sin(nx)}{\sin x}dx
\\
        &=O(1)+\int_{0}^{\pi / 4}\Bigl(\frac{\sin(t\sin x)}{\sin x}-\frac{t}{1+(t\sin x)^{2}}\Bigr)\frac{\sin(nx)}{\sin x}dx\\
  &\quad {} +\int_{0}^{\frac{\pi}{4}}\frac{t}{1+(t\sin x)^{2}}\frac{\sin(nx)}{\sin x}dx
\\
&=O(1)+U_{n}^{1}(t)+U_{n}^{2}(t),
\end{align*}
where $U^{1}(t)$ and $U_{n}^{2}(t)$ are correspondingly the first
and second integrals in the latter formula.

For $U_{n}^{1}(t)$ we have 
\begin{align*}
|U_{n}^{1}(t)|&\leqslant\int_{0}^{\pi / 4}\Bigl|\frac{\sin(t\sin x)}{\sin x}-\frac{t}{1+(t\sin x)^{2}}\Bigr|\frac{1}{\sin x}dx
\\
&=\int_{0}^{1 / \sqrt{2}}\Bigl|\frac{\sin(ty)}{y}-\frac{t}{1+(ty)^{2}}\Bigr| \frac{1}{y\sqrt{1-y^{2}}}dy\\
  &\leqslant\sqrt{2}\int_{0}^{1 / \sqrt{2}} \Bigl|\frac{\sin(ty)}{y}-\frac{t}{1+(ty)^{2}}\Bigr| \frac{1}{y}dy\\
&=\sqrt{2}t\int_{0}^{t / \sqrt{2}} \Bigl|\frac{\sin u}{u}-\frac{1}{1+u^{2}}\Bigr| \frac{1}{u}du\\
  &\leqslant\sqrt{2}t\int_{0}^{+\infty}\Bigl|\frac{\sin u}{u}-\frac{1}{1+u^{2}}\Bigr|\frac{1}{u}du.
\end{align*}
Convergence of the last integral follows because the integrand has the 
order $u^{-2}$ at infinity, and the order $u$ at zero. Thus we
get that 
\[
U_{n}^{1}(t)=O(t).
\]
Now let us estimate $U_{n}^{2}(t)$. We will need the known inequality
$\sin x\geqslant 2x / \pi $ that holds for any $x\in[0, \pi / 2]$.
We have 
\begin{align*}
  |U_{n}^{2}(t)|&\leqslant nt\int_{0}^{\pi / 4}\frac{t}{1+(t\sin x)^{2}}\frac{x}{\sin x}dx
                  \leqslant\pi nt\int_{0}^{\pi / 4}\frac{1}{1+4(tx)^{2} / \pi^{2}}dx\\
  &=\pi n\int_{0}^{t\pi / 4}\frac{1}{1+ 4u^{2} / \pi^{2} }du
\leqslant\pi n\int_{0}^{+\infty}\frac{1}{1+4 u^{2} / \pi^{2}}du.
\end{align*}
Consequently, 
\[
U_{n}^{2}(t)=O(n),
\]
\[
U_{n}(t)=O(1)+O(t)+O(n).
\]
Now consider in detail $K_{n}(t)$: 
\begin{align*}
  K_{n}(t)&=\int_{0}^{\pi / 4}\sin(t\sin x)\frac{\cos(nx)}{x}dx+O(1)\\
  &=\int_{0}^{\pi / 4}\sin(t\sin x)\frac{\cos(nx)-1}{x}dx
    +\int_{0}^{\pi / 4}\frac{\sin(t\sin x)}{x}dx+O(1)\\
  &=K_{n}^{1}(t)+K_{n}^{2}(t)+O(1),
\end{align*}
where $K_{n}^{1}(t)$ and $K_{n}^{2}(t)$ are correspondingly the
first and second integrals in the last formula. Let us show that $K_{n}^{2}(t)=O(1)$:
\begin{align*}
  K_{n}^{2}(t)&=\int_{0}^{\pi / 4}\frac{\sin(t\sin x)}{x}dx=\int_{0}^{\pi / 4}\frac{\sin(t\sin x)}{\sin x}dx+O(1)\\
  &=\int_{0}^{1 / \sqrt{2}}\frac{\sin(ty)}{y\sqrt{1-y^{2}}}dy+O(1)
    \\
&=\int_{0}^{1 / \sqrt{2}} \frac{\sin(ty)}{y}\Bigl(\frac{1}{\sqrt{1-y^{2}}}-1\Bigr)dy+\int_{0}^{1 / \sqrt{2}} \frac{\sin(ty)}{y}dy+O(1)=O(1).
\end{align*}
Now we integrate by parts $K_{n}^{1}(t)$: 
\begin{align*}
  K_{n}^{1}(t)&=\int\limits_{0}^{\pi / 4}\sin(t\sin x)\frac{\cos(nx)-1}{x}dx\\
  &=-\frac{\cos(t\sin x)}{t\cos x}\frac{\cos(nx)-1}{x}\Big|_{0}^{\pi / 4}
    +\frac{1}{t}\int\limits_{0}^{\pi / 4}\cos(t\sin x)\frac{d}{dx}\frac{\cos(nx)-1}{x\cos x}dx\\
  &=O(1)-\frac{n}{t}\int\limits_{0}^{\pi / 4}\cos(t\sin x)\frac{\sin(nx)}{x\cos x}dx
    -\frac{1}{t}\int\limits_{0}^{\pi / 4}\cos(t\sin x)\frac{\cos(nx)-1}{x^{2}\cos x}dx\\
  &\quad {} +\frac{1}{t}\int\limits_{0}^{\pi / 4}\cos(t\sin x)\frac{\cos(nx)-1}{x\cos^{2}x}\sin xdx.
\end{align*}
The first integral in the last expression equals $\tilde{C}_{n}(t)+O(1)$,
the third integral has order $O(1)$. Thus 
\[
K_{n}^{1}(t)=O(1)-\frac{n}{t}\tilde{C}_{n}(t)-\frac{1}{t}Z_{n}(t),
\]
where 
\[
Z_{n}(t)=\int_{0}^{\pi / 4}\cos(t\sin x)\frac{\cos(nx)-1}{x^{2}}dx.
\]
Note that the number $n>0$ could be considered as an arbitrary real number because
we never used the fact that $n$ is an integer. Then 
\[
\frac{d}{dn}Z_{n}(t)=-\int_{0}^{\pi / 4}\cos(t\sin x)\frac{\sin(nx)}{x}dx=-\tilde{C}_{n}(t)
\]
and 
\[
Z_{n}(t)=O(1)-\int_{0}^{n}\tilde{C}_{m}(t)dm.
\]
Finally we get 
\[
K_{n}(t)=O(1)-\frac{n}{t}\tilde{C}_{n}(t)+\frac{1}{t}\int_{0}^{n}\tilde{C}_{m}(t)dm.
\]
Substitute formulas for $K_{n}(t)$ and $U_{n}(t)$ to the formula
(\ref{tildecnt}) for $\tilde{C}_{n}(t)$: 
\[
\tilde{C}_{n}(t)=O(1)+\left(\frac{n}{t}\right)^{2}\tilde{C}_{n}(t)-\frac{n}{t^{2}}\int_{0}^{n}\tilde{C}_{m}(t)dm.
\]
It follows that for $t / n \geqslant\gamma_{2}>1$, 
\begin{align*}
  |\tilde{C}_{n}(t)|&=\Biggl|\frac{1}{1-(n / t)^{2}}\Biggl(O(1)-\frac{n}{t^{2}}\int_{0}^{n}\tilde{C}_{m}(t)dm\Biggr)\Biggr| \\
&  \leqslant\frac{1}{1- \gamma_{2}^{-2}}\Biggl(O(1)+\frac{1}{t}\int_{0}^{n}|\tilde{C}_{m}(t)|dm\Biggr).
\end{align*}
Gronwall's lemma gives 
\[
  |\tilde{C}_{n}(t)|\leqslant\frac{c}{1- \gamma_{2}^{-2}}
  \exp\Bigl(\frac{1}{1- \gamma_{2}^{-2}}\frac{n}{t}\Bigr)\leqslant\frac{c\gamma_{2}^{2}}{\gamma_{2}^{2}-1}\exp\Bigl(\frac{\gamma_{2}}{\gamma_{2}^{2}-1}\Bigr)
\]
for some constant $c>0$ not depending on $n,t,\gamma_{2}$. Lemma
\ref{greaterOnelemma} is proved.
\end{proof}

\subsection{Uniform boundedness for $\boldsymbol{\gamma_{1}<\gamma_{n}(t)<\gamma_{2}}$. Preliminary
lemmas}

Using the formula for the product $\sin \, \cos$ we
can rewrite the formula for $\tilde{C}_{n}(t)$ as 
\begin{align*}
  \tilde{C}_{n}(t)&=\frac{1}{2}\Biggl(\int_{0}^{\pi / 4}\frac{\sin(nx+t\sin x)}{x}dx+\int_{0}^{\pi / 4}\frac{\sin(nx-t\sin x)}{x}dx\Biggr)\\
&  =\frac{1}{2}\bigl(R_{n}(t)+M_{n}(t)\bigr),
\end{align*}
where 
\[
R_{n}(t)=\int_{0}^{\pi / 4}\frac{\sin(nx+t\sin x)}{x}dx,\quad M_{n}(t)=\int_{0}^{\pi / 4}\frac{\sin(nx-t\sin x)}{x}dx.
\]
\begin{lemma} \label{rntemestimate} There exists a constant $\alpha>0$
such that for all $t\geqslant0$ and real $n>0$ the following inequality
holds 
\[
|R_{n}(t)|\leqslant\alpha.
\]
\end{lemma}
\begin{proof} Denote 
$
f(x)=nx+t\sin x
$
and note that 
\[
f'(x)=\frac{d}{dx}f(x)=n+t\cos x>0
\]
for $x\in[0,\pi/4]$. Denote 
\[
a=f\Bigl(\frac{\pi}{4}\Bigr)=\frac{\pi n}{4}+\frac{t}{\sqrt{2}}
\]
the maximal value of $f(x)$ on $[0,\pi / 4]$. Since $f(x)$ is monotone
increasing on $[0, \pi / 4]$, there exists a smooth function $\varphi:[0,a]\to[0,\pi/4]$
such that 
\[
f(\varphi(u))=u
\]
for all $u\in[0,a]$. Thus we can change variables $x=\varphi(u)$
in the integral $R_{n}(t)$: 
\[
R_{n}(t)=\int_{0}^{a}\frac{\sin u}{\varphi(u)}\varphi'(u)du=\int_{0}^{a}\frac{\sin u}{u}h(u)du,
\]
where we denoted 
\[
h(u)=\frac{u}{\varphi(u)}\varphi'(u).
\]
It is obvious that $h(0)=1$. Then 
\begin{align}
  R_{n}(t)&=\int_{0}^{a}\frac{\sin u}{u}(h(u)-h(0))du+h(0)\int_{0}^{a}\frac{\sin u}{u}du \nonumber \\
  &=\int_{0}^{a}\frac{h(u)-h(0)}{u}\sin u\ du+\mathrm{Si}(a),\label{Rnteq}
\end{align}
where 
\[
\mathrm{Si}(x)=\int_{0}^{x}\frac{\sin u}{u}du
\]
is the integral sine. It is well known that $\lim_{x\rightarrow+\infty}\mathrm{Si}(x)=\pi / 2$.
Also, since $\mathrm{Si}(x)$ is continuous, it is bounded on $[0,+\infty)$.

Let us estimate the fractional term in the first integral containing
$h(u)$. We want to use Lemma \ref{techLemmah}. For this we need
the following estimates of $\varphi(u)$: 
\[
\varphi'(u)=\frac{1}{f'(\varphi(u))}=\frac{1}{n+t\cos\varphi(u)}\leqslant\frac{1}{n+\frac{t}{\sqrt{2}}}\leqslant\frac{\sqrt{2}}{n+t}.
\]
The lower bound is evident: 
\[
\varphi'(u)\geqslant\frac{1}{n+t}.
\]
Let us estimate the second derivative 
\[
\varphi''(u)=-\frac{f''(\varphi(u))\varphi'(u)}{(f'(\varphi(u)))^{2}}=\frac{t\sin\varphi(u)}{(f'(\varphi(u)))^{3}}\leqslant\frac{t}{(n+\frac{t}{\sqrt{2}})^{3}}\leqslant(\sqrt{2})^{3}\frac{t}{(n+t)^{3}}.
\]
Then from Lemma \ref{techLemmah} we get the following inequality: 
\[
\Bigl|\frac{h(u)-h(0)}{u}\Bigr|\leqslant4\frac{t}{(n+t)^{3}}\frac{3 / (n+t)}{(n+t)^{-2}}=\frac{12t}{(n+t)^{2}}.
\]
Hence the integral in  (\ref{Rnteq}) can be estimated
as follows: 
\[
\Biggl|\int_{0}^{a}\frac{h(u)-h(0)}{u}\sin u\ du\Biggr| \leqslant\frac{12at}{(n+t)^{2}}=12t\frac{\pi n / 4 + t / \sqrt{2}}{(n+t)^{2}}\leqslant\frac{12t}{n+t}\leqslant12.
\]
Thereby Lemma \ref{rntemestimate} is proved.
\end{proof}

\begin{lemma} \label{techLemmah} Let a real function $\phi(u)\in C^{2}([0,a])$
be given, $ a>0$. Assume also that for some positive constants $c_{1},c_{2},c_{3}$
and all $u\in[0,a]$ we have: 
\[
0<c_{1}\leqslant\phi'(u)\leqslant c_{2},
\]
\[
|\phi''(u)|\leqslant c_{3},
\]
and also $\phi(0)=0$. Then for the function 
\[
h(u)=\frac{u}{\phi(u)}\phi'(u)
\]
we have $h(0)=1$, and for all $u\in[0,a]$ the following inequality
holds 
\[
\Bigl|\frac{h(u)-h(0)}{u}\Bigr|\leqslant c_{3}\frac{c_{1}+c_{2}}{c_{1}^{2}}.
\]
\end{lemma}
\begin{proof}
  Equality $h(0)=1$ follows immediately. By Lagrange theorem
we have 
\[
\Bigl|\frac{h(u)-h(0)}{u}\Bigr|\leqslant\max_{v\in[0,a]}|h'(v)|.
\]
Now let us estimate the derivative $h'(u)$. Denote 
\[
\psi(u)=\frac{\phi(u)}{u}.
\]
Then $h(u)=\phi'(u) / \psi(u)$ and 
\begin{equation}
h'(u)=\frac{\phi''(u)}{\psi(u)}-\frac{\phi'(u)\psi'(u)}{\psi^{2}(u)}.\label{hderfortemp}
\end{equation}
Let us estimate now $\psi(u)$: 
\[
\psi(u)=\frac{\phi(u)}{u}\geqslant\min_{v\in[0,a]}\phi'(v)\geqslant c_{1}>0.
\]
Again by Lagrange theorem we have 
\[
\psi'(u)=\frac{u\phi'(u)-\phi(u)}{u^{2}}=\frac{\phi'(u)- \phi(u) / u}{u}=\frac{\phi'(u)-\phi'(\theta(u))}{u}
\]
for some $\theta(u)\in[0,u]$. Continue the previous equality:
\[
|\psi'(u)|\leqslant\max_{v\in[0,a]}|\phi''(v)|\frac{u-\theta(u)}{u}\leqslant c_{3}.
\]
Substitute the obtained estimates to (\ref{hderfortemp}): 
\[
|h'(u)|\leqslant\frac{c_{3}}{c_{1}}+\frac{c_{2}c_{3}}{c_{1}^{2}}=c_{3}\frac{c_{1}+c_{2}}{c_{1}^{2}}.
\]
Lemma \ref{techLemmah} is proved.
\end{proof}

Now we shall estimate the integral $M_{n}(t)$ for the remaining domain of
parameters $\gamma_{1}<\gamma_{n}(t)<\gamma_{2}$. Put 
\[
t=(1+\varepsilon)n
\]
for some $|\varepsilon|<\varepsilon'<1$. Then 
\begin{align*}
  M_{n}(t)&=M_{n}((1+\varepsilon)n)=L_{n}(\varepsilon)=\int_{0}^{\pi / 4}\frac{\sin(n(x-(1+\varepsilon)\sin x)}{x}dx\\
  &=
  \int_{0}^{\pi / 4}\frac{\sin(nf_{\varepsilon}(x))}{x}dx,
\end{align*}
where 
\[
f_{\varepsilon}(x)=x-(1+\varepsilon)\sin x.
\]

\subsection{Uniform boundedness for $\boldsymbol{\gamma_{1}<\gamma_{n}(t)\leqslant1}$}

\begin{lemma} \label{gammalessonelemma} For any $0<\varepsilon'<1$
there exists a constant $c=c(\varepsilon')>0$ such that for all $\varepsilon\in[-\varepsilon';0]$
and $n\geqslant1$ the following inequality holds: 
\[
|L_{n}(\varepsilon)|\leqslant c.
\]
\end{lemma}
\begin{proof} We have 
\[
f'_{\varepsilon}(x)=\frac{d}{dx}f_{\varepsilon}(x)=1-(1+\varepsilon)\cos x.
\]
It follows that $f'_{\varepsilon}(x)\geqslant0$ for $x\in[0,\pi/4],$
 and moreover $f'_{\varepsilon}(x)=0$ only for $\varepsilon=0,x=0$.
Then the function $f_{\varepsilon}(x)$ is monotone increasing for
$x\in[0,\pi/4]$ and there exists a monotone continuous function 
\[
\varphi_{\varepsilon}(u):[0,\delta_{\varepsilon}]\to[0,\pi/4],\ 
\]
such that 
\[
f_{\varepsilon}(\varphi_{\varepsilon}(u))=-\varepsilon u+\frac{u^{3}}{6}=:g_{\varepsilon}(u)
\]
and $\delta_{\varepsilon}=g_{\varepsilon}^{-1}(f_{\varepsilon}(\pi / 4 ))$,
where $g_{\varepsilon}^{-1}$ is the inverse function, existing due
to monotonicity of $g_{\varepsilon}$. Indeed, $\varphi_{\varepsilon}(u)$
can be written as $\varphi_{\varepsilon}(u)=f_{\varepsilon}^{-1}(g_{\varepsilon}(u))$.

Let us show that for all $\varepsilon\in[-\varepsilon',0]$ the function
$\varphi_{\varepsilon}(u)$ is continuously differentiable in $u$
on the interval $[0,\delta_{\varepsilon}]$, that is $\varphi_{\varepsilon}(u)\in C^{1}([0,\delta_{\varepsilon}])$.
If $\varepsilon<0$, then $f'_{\varepsilon}(x)>0$ for all $x\in[0,\pi/4]$.
Then by classical inverse function theorems it follows that $\varphi_{\varepsilon}(u)\in C^{1}([0,\delta_{\varepsilon}])$.
For $\varepsilon=0$ and $u\in(0,\delta_{\varepsilon}]$ by the same
reasons we have continuous differentiability of $\varphi_{\varepsilon}(u)$
for $0<u\leqslant\delta_{\varepsilon}$. Let us show that there exists
derivative of $\varphi_{0}(u)$ at $u=0$. From Taylor formula we
have: 
\[
f_{0}(x)=\frac{x^{3}}{6}+\bar{\bar{o}}(x)\quad \mbox{as}\ x\rightarrow0.
\]
It is clear that $\varphi_{0}(0)=0$ and $\varphi_{0}(u)\rightarrow0$
as $u\rightarrow0+$. Then 
\[
f_{0}(\varphi_{0}(u))=\frac{\varphi_{0}^{3}(u)}{6}+\bar{\bar{o}}(\varphi_{0}(u))=\frac{u^{3}}{6}.
\]
Finally we get 
\[
\lim_{u\rightarrow0+}\frac{\varphi_{0}(u)}{u}=1.
\]
Thereby we proved that $\varphi_{\varepsilon}(u)\in C^{1}([0,\delta_{\varepsilon}])$
for all $\varepsilon\in[-\varepsilon',0]$. We shall see further on
that in fact the function $\varphi_{\varepsilon}(u)$ is ``close''
to $u$. 

Now use the change of variables $x=\varphi_{\varepsilon}(u)$ in the
integral for $L_{n}(\varepsilon)$: 
\[
L_{n}(\varepsilon)=\int_{0}^{\delta_{\varepsilon}}\frac{\sin(ng_{\varepsilon}(u))}{\varphi_{\varepsilon}(u)}\varphi'_{\varepsilon}(u)du=\int_{0}^{\delta_{\varepsilon}}\frac{\sin(ng_{\varepsilon}(u))}{u}h_{\varepsilon}(u)du,
\]
where 
\[
h_{\varepsilon}(u)=\frac{u}{\varphi_{\varepsilon}(u)}\varphi'_{\varepsilon}(u).
\]
Obviously $h_{\varepsilon}(0)=1$. Hence
\begin{equation}
L_{n}(\varepsilon)=\int_{0}^{\delta_{\varepsilon}}\frac{\sin(ng_{\varepsilon}(u))}{u}du+\int_{0}^{\delta_{\varepsilon}}\sin(ng_{\varepsilon}(u))\frac{h_{\varepsilon}(u)-h_{\varepsilon}(0)}{u}du.\label{Lnepslesszerofor}
\end{equation}
From Lemmas \ref{epslesszerosubstit} and \ref{techLemmah} we have
the estimate for the second interval: 
\[
\Biggl|\int_{0}^{\delta_{\varepsilon}}\sin(ng_{\varepsilon}(u))\frac{h_{\varepsilon}(u)-h_{\varepsilon}(0)}{u}du\Biggr|\leqslant c_{5}c_{3}\frac{c_{1}+c_{2}}{c_{1}^{2}},
\]
where the constants $c_{k},\ k=1,\ldots,5$ were defined in Lemma
\ref{epslesszerosubstit}.

Now let us estimate the first integral in (\ref{Lnepslesszerofor}):
\begin{align*}
  \tilde{L}_{n}(\varepsilon)&=\int_{0}^{\delta_{\varepsilon}}\frac{\sin(ng_{\varepsilon}(u))}{u}du=
                              \int_{0}^{\delta_{\varepsilon}}\frac{\sin(n(-\varepsilon u+ u^{3} / 6))}{u}du\\
  &=\int_{0}^{\delta_{\varepsilon}n}\frac{\sin(-\varepsilon y+ y^{3} / (6n^{2}))}{y}dy.
\end{align*}
In the last integral the cubic term under the sine, for $y<n^{2/3}$,
has at most the same order as the linear term. Taking into account
this observation, subdivide this integral in two: 
\begin{align*}
  \tilde{L}_{n}(\varepsilon)&=\int_{0}^{\delta_{\varepsilon}n^{2/3}}\frac{\sin(-\varepsilon y+ y^{3} / (6n^{2}))}{y}dy+
                              \int_{\delta_{\varepsilon}n^{2/3}}^{\delta_{\varepsilon}n}\frac{\sin(-\varepsilon y+ y^{3} /(6n^{2}))}{y}dy\\
  &=\tilde{L}_{n}^{1}(\varepsilon)+\tilde{L}_{n}^{2}(\varepsilon),
\end{align*}
where we denoted by $\tilde{L}_{n}^{1}(\varepsilon),\tilde{L}_{n}^{2}(\varepsilon)$
the first and second integrals in the latter formula. Let us estimate
firstly $\tilde{L}_{n}^{1}(\varepsilon)$. By Taylor formula for the
sine for all $y,z\geqslant0$ we have the equality 
\[
\sin(-\varepsilon y+z)=\sin(-\varepsilon y)+z\cos(-\varepsilon y+\theta(z))
\]
for some $\theta(z)\in[0,z]$. Then 
\begin{align*}
  \tilde{L}_{n}^{1}(\varepsilon)&=\int_{0}^{\delta_{\varepsilon}n^{2/3}}\frac{\sin(-\varepsilon y+ y^{3} / (6n^{2}))}{y}du \\
&  =\int_{0}^{\delta_{\varepsilon}n^{2/3}}\frac{\sin(-\varepsilon y)}{y}du+\int_{0}^{\delta_{\varepsilon}n^{2/3}}
  \frac{y^{3}}{6n^{2}}\frac{\cos\bigl(-\varepsilon y+\theta( y^{3} / (6n^{2}))\bigr)}{y}dy.
\end{align*}
The first integral is evidently equal to $\mathrm{Si}(-\varepsilon\delta_{\varepsilon}n^{2/3})$
and thus is bounded uniformly in $n\geqslant0$ and $\epsilon\leqslant0$.
For the second integral we have the bounds 
\begin{align*}
  \Biggl|\int_{0}^{\delta_{\varepsilon}n^{2/3}}\frac{y^{3}}{6n^{2}}\frac{\cos\bigl(-\varepsilon y+\theta(y^{3} / (6n^{2}))\bigr)}{y}dy\Biggr|
  &\leqslant\frac{1}{6n^{2}}\int_{0}^{\delta_{\varepsilon}n^{2/3}}y^{2}\ dy \\
  &=\frac{1}{18n^{2}}(\delta_{\varepsilon}n^{2/3})^{3}=\frac{\delta_{\varepsilon}^{3}}{18}.
\end{align*}
Then by Lemma \ref{epslesszerosubstit} this integral is also uniformly
bounded. Thereby we have shown that 
\[
\tilde{L}_{n}^{1}(\varepsilon)=O(1)
\]
uniformly in $n\geqslant0$ and $\epsilon\leqslant0$. Now estimate
$\tilde{L}_{n}^{2}(\varepsilon)$. Integration by parts gives: 
\begin{align*}
  \tilde{L}_{n}^{2}(\varepsilon)&=\int_{\delta_{\varepsilon}n^{2/3}}^{\delta_{\varepsilon}n}\frac{\sin(-\varepsilon y+ y^{3} / (6n^{2}))}{y}dy\\
                                  &=-\frac{\cos(-\varepsilon y+ y^{3} /(6n^{2}))}{(-\varepsilon+y^{2} / (2n^{2}))y}\Big|_{\delta_{\varepsilon}n^{2/3}}^{\delta_{\varepsilon}n}
\\
                                &\quad {} +\int_{\delta_{\varepsilon}n^{2/3}}^{\delta_{\varepsilon}n}\cos\Bigl(-\varepsilon y+\frac{y^{3}}{6n^{2}}\Bigr)
                                  \frac{d}{dy}\frac{1}{(-\varepsilon+ y^{2} / (2n^{2}))y}dy.
\end{align*}
The first term after substitution, by Lemma \ref{epslesszerosubstit} and the 
lower estimate for $\delta_{\varepsilon}$, equals $O(1)$. Also 
\[
  \Bigl|\frac{d}{dy}\frac{1}{(-\varepsilon+y^{2} / (2n^{2}))y}\Bigr|=
  \frac{1}{(-\varepsilon+y^{2} / (2n^{2}))y^{2}}+\frac{1}{n^{2}}\frac{1}{(-\varepsilon+ y^{2} / (2n^{2}))^{2}}\leqslant\frac{6n^{2}}{y^{4}}
\]
and thus 
\[
  |\tilde{L}_{n}^{2}(\varepsilon)|\leqslant O(1)+6n^{2}\int_{\delta_{\varepsilon}n^{2/3}}^{\delta_{\varepsilon}n}\frac{1}{y^{4}}dy
  =O(1)+2n^{2}\Bigl(\frac{1}{(\delta_{\varepsilon}n^{2/3})^{3}}-\frac{1}{(\delta_{\varepsilon}n)^{3}}\Bigr)=O(1).
\]
Thus, Lemma \ref{gammalessonelemma} is proved.
\end{proof}

\begin{lemma} \label{epslesszerosubstit} For any $\varepsilon'<1$
there exist positive constants $c_{1},c_{2},c_{3},c_{4},c_{5}$ such
that for all $\varepsilon\in[-\varepsilon';0]$ and any $u\in[0,\delta_{\varepsilon}]$
the following inequalities hold: 
\[
0<c_{1}\leqslant\varphi_{\varepsilon}'(u)\leqslant c_{2},
\]
\[
|\varphi_{\varepsilon}''(u)|\leqslant c_{3},
\qquad 
c_{4}\leqslant\delta_{\varepsilon}\leqslant c_{5}.
\]
\end{lemma}
\begin{proof}
  The proof consists of several parts.

  1. Firstly, we shall prove the inequality: 
\begin{equation}
u\leqslant\varphi_{\varepsilon}(u)\leqslant\frac{\pi}{2}u \label{varepsmajor}
\end{equation}
that holds for any $u\in[0,\delta_{\varepsilon}]$. Note first that
for all $x\in[0,\pi/4]$ the following inequalities hold 
\[
\frac{1}{3\pi}x^{3}\leqslant x-\sin x\leqslant\frac{x^{3}}{6}.
\]
Hence, for the function $f_{\varepsilon}(x)=x-(1+\varepsilon)\sin x=x-\sin x-\varepsilon\sin x$
we have: 
\[
\frac{1}{3\pi}x^{3}-\varepsilon\frac{2}{\pi}x\leqslant f_{\varepsilon}(x)\leqslant\frac{x^{3}}{6}-\varepsilon x=g_{\varepsilon}(x).
\]
The left part can be estimated as follows 
\[
  \frac{1}{3\pi}x^{3}-\varepsilon\frac{2}{\pi}x=\frac{2}{\pi}\frac{x^{3}}{6}-\varepsilon\frac{2}{\pi}x\geqslant\Bigl(\frac{2}{\pi}\Bigr)^{3}
  \frac{x^{3}}{6}-\varepsilon\frac{2}{\pi}x=g_{\varepsilon}\Bigl(\frac{2}{\pi}x\Bigr)
\]
and we get the bounds 
\[
g_{\varepsilon}\Bigl(\frac{2}{\pi}x\Bigr)\leqslant f_{\varepsilon}(x)\leqslant g_{\varepsilon}(x).
\]
After substitution $x=\varphi_{\varepsilon}(u)$ we have 
\[
g_{\varepsilon}\Bigl(\frac{2}{\pi}\varphi_{\varepsilon}(u)\Bigr)\leqslant f_{\varepsilon}(\varphi_{\varepsilon}(u))=g_{\varepsilon}(u)\leqslant g_{\varepsilon}(\varphi_{\varepsilon}(u)).
\]
From this, taking into account that $g_{\varepsilon}(u)$ is monotone
increasing, we obtain (\ref{varepsmajor}).

Taylor expansion gives for any $x\geqslant0$: 
\begin{equation}
f_{\varepsilon}(x)=-\varepsilon x+\frac{x^{3}}{6}+R_{\varepsilon}(x)=g_{\varepsilon}(x)+R_{\varepsilon}(x),\label{fepsTailorExp}
\end{equation}
where 
\[
R_{\varepsilon}(x)=\frac{1}{4!}\int_{0}^{x}(x-s)^{4}f_{\varepsilon}^{(5)}(s)\ ds=\frac{1}{5!}x^{5}f_{\varepsilon}^{(5)}(\theta_{\varepsilon}(x))=-(1+\epsilon)\frac{1}{5!}x^{5}\cos(\theta_{\varepsilon}(x))
\]
for some $\theta_{\varepsilon}(x)\in[0,x]$. Putting $x=\varphi_{\varepsilon}(u)$
in the expansion (\ref{fepsTailorExp}), we get 
\[
g_{\varepsilon}(u)=g_{\varepsilon}(\varphi_{\varepsilon}(u))+R_{\varepsilon}(\varphi_{\varepsilon}(u)).
\]
The tangent line equation is 
\[
y(v)=g_{\varepsilon}(u)+g_{\varepsilon}'(u)(v-u).
\]
Using convexity property of function $g_{\varepsilon}(u)$ we have
\[
g_{\varepsilon}(\varphi_{\varepsilon}(u))\geqslant y(\varphi_{\varepsilon}(u))
\]
and 
\begin{align*}
  \varphi_{\varepsilon}(u)-u&\leqslant\frac{g_{\varepsilon}(\varphi_{\varepsilon}(u))-g_{\varepsilon}(u)}{g'_{\varepsilon}(u)}
  =-\frac{R_{\varepsilon}(\varphi_{\varepsilon}(u))}{-\varepsilon+ u^{2} / 2}\\
                            &\leqslant\frac{1}{5!}\frac{\varphi_{\varepsilon}^{5}(u)}{(-\varepsilon+ u^{2} / 2)}
                              \leqslant\frac{2^{5}}{5!}\frac{u^{5}}{(-\varepsilon+u^{2} / 2)}\leqslant\frac{2^{5}}{5!}2u^{3}.
\end{align*}
Thus, we have proved that 
\begin{equation}
0\leqslant\varphi_{\varepsilon}(u)-u\leqslant cu^{3}\label{phiumaj}
\end{equation}
for some absolute constant $c$ not depending on $\varepsilon$ and
$u$. Also, the bound for $\delta_{\varepsilon}$ follows from this.
Indeed, by definition we have $\varphi_{\varepsilon}(\delta_{\varepsilon})= \pi / 4$.
Then, 
\begin{equation}
\delta_{\varepsilon}\leqslant\frac{\pi}{4}.\label{deltaepsineq}
\end{equation}
From  (\ref{phiumaj}) we get that 
\[
c\delta_{\varepsilon}^{3}+\delta_{\varepsilon}\geqslant\frac{\pi}{4}
\]
and that $\delta_{\varepsilon}>c'$ for a constant $c'$ not depending
on $\varepsilon$.

\medskip

2. Now we shall estimate the derivative of the function $\varphi_{\varepsilon}(u)$.
Note that the following equality holds true:
\begin{equation}
\varphi'_{\varepsilon}(u)=\frac{g'_{\varepsilon}(u)}{f'_{\varepsilon}(\varphi_{\varepsilon}(u))}.\label{varphidefrav}
\end{equation}
Since 
\[
f'_{\varepsilon}(x)=1-(1+\varepsilon)\cos x=1-\cos x-\varepsilon\cos x\leqslant\frac{x^{2}}{2}-\epsilon,
\]
 using (\ref{varepsmajor}), we get 
\[
  \varphi'_{\varepsilon}(u)\geqslant\frac{g'_{\varepsilon}(u)}{\varphi_{\varepsilon}^{2}(u)  / 2-\epsilon}
  =\frac{-\varepsilon+u^{2} / 2}{\varphi_{\varepsilon}^{2}(u) / 2-\epsilon}
  \geqslant\frac{-\varepsilon+ u^{2} / 2}{ (\pi / 2)^{2} u^{2} / 2-\epsilon}
  \geqslant\frac{1}{(\pi / 2)^{2}}=\frac{4}{\pi^{2}}.
\]
To estimate the second derivative of $\varphi_{\varepsilon}(u)$ we
will need more exact estimates of the difference $\varphi'_{\varepsilon}(u)-1$.
So, we need inequalities for this difference. By (\ref{varphidefrav})
we have 
\begin{equation}
\varphi'_{\varepsilon}(u)-1=\frac{g'_{\varepsilon}(u)-f'_{\varepsilon}(\varphi_{\varepsilon}(u))}{f'_{\varepsilon}(\varphi_{\varepsilon}(u))}.\label{phiespdereq}
\end{equation}
From 
\[
f'_{\varepsilon}(x)=1-(1+\varepsilon)\cos x=1-\cos x-\varepsilon\cos x\geqslant1-\cos x\geqslant\frac{1}{\pi}x^{2}
\]
one can get 
\begin{equation}
f'_{\varepsilon}(\varphi_{\varepsilon}(u))\geqslant\frac{1}{\pi}\varphi_{\varepsilon}^{2}(u)\geqslant\frac{1}{\pi}u^{2}.\label{fepsderminor}
\end{equation}
And from (\ref{fepsTailorExp}) we have 
\[
f'_{\varepsilon}(x)=g'_{\varepsilon}(x)+R'_{\varepsilon}(x).
\]
Moreover 
\[
R'_{\varepsilon}(x)=\frac{1}{4!}x^{4}f_{\varepsilon}^{(5)}(\eta_{\varepsilon}(x))=-(1+\epsilon)\frac{1}{4!}x^{4}\cos(\eta_{\varepsilon}(x))
\]
for some $\eta_{\varepsilon}(x)\in[0,x]$. Finally, 
\[
f'_{\varepsilon}(\varphi_{\varepsilon}(u))=g'_{\varepsilon}(\varphi_{\varepsilon}(u))+R'_{\varepsilon}(\varphi_{\varepsilon}(u)).
\]
Again using Taylor expansion for $g'(\varphi_{\varepsilon}(u))$ at
 $u$, we get: 
\[
g'_{\varepsilon}(\varphi_{\varepsilon}(u))=g'_{\varepsilon}(u)+g''_{\varepsilon}(u)(\varphi_{\varepsilon}(u)-u)+\frac{g_{\varepsilon}^{(3)}(u)}{2}(\varphi_{\varepsilon}(u)-u)^{2}.
\]
And from (\ref{phiespdereq}), using (\ref{varepsmajor}) and (\ref{phiumaj}),
we get 
\begin{align*}
  |\varphi'_{\varepsilon}(u)-1|&=\Bigl|\frac{g'_{\varepsilon}(u)-g'(\varphi_{\varepsilon}(u))-R'_{\varepsilon}(\varphi_{\varepsilon}(u))}{f'_{\varepsilon}(\varphi_{\varepsilon}(u))}\Bigr|
  \\
  &\leqslant\frac{|g'_{\varepsilon}(u)-g'(\varphi_{\varepsilon}(u))|}{u^{2} / \pi} +\frac{\varphi_{\varepsilon}^{4}(u)}{4! \, u^{2} / \pi}
\\
&\leqslant\frac{u(\varphi_{\varepsilon}(u)-u)+(\varphi_{\varepsilon}(u)-u)^{2} / 2}{u^{2} / \pi} +\frac{\pi}{4!}\Bigl(\frac{\pi}{2}\Bigr)^{4}u^{2}\leqslant cu^{2}
\end{align*}
for some constant $c>0$ not depending on $\varepsilon$ and $u$.
Thus, we have proved 
\begin{equation}
|\varphi'_{\varepsilon}(u)-1|\leqslant cu^{2}.\label{derphionecmp}
\end{equation}

\medskip

3. To estimate the second derivative of $\varphi_{\varepsilon}(u)$ we
use the equality 
\begin{align*}
  \varphi''_{\varepsilon}(u)&=\frac{g''_{\varepsilon}(u)}{f'_{\varepsilon}(\varphi_{\varepsilon}(u))}-\frac{g'_{\varepsilon}(u)f''_{\varepsilon}(\varphi_{\varepsilon}(u))\varphi'_{\varepsilon}(u)}{(f'_{\varepsilon}(\varphi_{\varepsilon}(u)))^{2}}\\
  &=\frac{g''_{\varepsilon}(u)}{f'_{\varepsilon}(\varphi_{\varepsilon}(u))}-\frac{f''_{\varepsilon}(\varphi_{\varepsilon}(u))(\varphi'_{\varepsilon}(u))^{2}}{f'_{\varepsilon}(\varphi_{\varepsilon}(u)))}
\\
&=\frac{g''_{\varepsilon}(u)-f''_{\varepsilon}(\varphi_{\varepsilon}(u))(\varphi'_{\varepsilon}(u))^{2}}{f'_{\varepsilon}(\varphi_{\varepsilon}(u))}.
\end{align*}
From (\ref{fepsTailorExp}) we have 
\[
f''_{\varepsilon}(x)=g''_{\varepsilon}(x)+R''_{\varepsilon}(x)
\]
and moreover, 
\[
R''_{\varepsilon}(x)=\frac{1}{3!}x^{3}f_{\varepsilon}^{(5)}(\xi_{\varepsilon}(x))=-(1+\epsilon)\frac{1}{3!}x^{3}\cos(\xi_{\varepsilon}(x))
\]
for some $\xi_{\varepsilon}(x)\in[0,x]$. Hence,
\[
f''_{\varepsilon}(\varphi_{\varepsilon}(u))=g''_{\varepsilon}(\varphi_{\varepsilon}(u))+R''_{\varepsilon}(\varphi_{\varepsilon}(u)).
\]
Inequalities (\ref{derphionecmp}) and (\ref{phiumaj}) can be rewritten
in terms of $\Delta=\varphi_{\varepsilon}(u)-u$ as follows: 
\begin{equation}
0\leqslant\Delta\leqslant cu^{3},\ \quad|\Delta'|\leqslant cu^{2}. \label{MainIneqDeltaSense}
\end{equation}
Then 
\[
g''_{\varepsilon}(\varphi_{\varepsilon}(u))=g''_{\varepsilon}(u)+(\varphi_{\varepsilon}(u)-u)=g''_{\varepsilon}(u)+\Delta.
\]
Thus, from (\ref{fepsderminor}), (\ref{MainIneqDeltaSense}), (\ref{deltaepsineq})
we get: 
\begin{align*}
  |\varphi''_{\varepsilon}(u)|&=\Bigl|\frac{g''_{\varepsilon}(u)-g''_{\varepsilon}(\varphi_{\varepsilon}(u))(\varphi'_{\varepsilon}(u))^{2}-R''_{\varepsilon}(\varphi_{\varepsilon}(u))(\varphi'_{\varepsilon}(u))^{2}}{f'_{\varepsilon}(\varphi_{\varepsilon}(u)))}\Bigr| \\
  &\leqslant
    \frac{|g''_{\varepsilon}(u)-g''_{\varepsilon}(\varphi_{\varepsilon}(u))(\Delta'+1)^{2}|}{u^{2} / \pi} +\frac{\varphi_{\varepsilon}^{3}(u)}{3! \, u^{2} / \pi}(1+cu^{2})^{2}\\
  &\leqslant g''(u)\frac{|1-(\Delta'+1)^{2}|}{u^{2} / \pi} +\Delta\frac{(\Delta'+1)^{2}}{u^{2} / \pi }+
    +\Bigl(\frac{\pi}{2}\Bigr)^{3}\frac{u^{3}}{3! \, u^{2} / \pi}(1+cu^{2})^{2}\\
                              &\leqslant u\frac{|\Delta'||\Delta'+2|}{u^{2} / \pi}+\frac{cu^{3}(1+cu^{2})^{2}}{u^{2} / \pi}+\Bigl(\frac{\pi}{2}\Bigr)^{3}
                                \frac{u}{3! / \pi}(1+cu^{2})^{2}\\
  &\leqslant c'u
\end{align*}
for some constant $c'$ not depending on $c$.
Lemma \ref{epslesszerosubstit} is completely proved.
\end{proof}

\subsection{Uniform boundedness for $\boldsymbol{1<\gamma_{n}(t)<\gamma_{2}}$}

This case is equivalent to the assumption that $0<\varepsilon<\eps'$
for some $\eps'>0$.

\begin{lemma}\label{gammagreateronelemma} There exist $\varepsilon'>0$
and  $c>0$ such that for all $\varepsilon\in(0;\varepsilon']$
and all $n\geqslant1$,
$
|L_{n}(\varepsilon)|\leqslant c.
$
\end{lemma}
\begin{proof}
  The idea and scheme of the proof are the same
as for Lemma \ref{gammalessonelemma} -- to do appropriate change
of variables and then estimate the simpler integral. The only difference
will be that the phase function $f_{\varepsilon}(x)=x-(1+\varepsilon)\sin x$
has critical point on the integration interval. Now we pass to detailed
proof.

The derivative of the phase function 
\[
f'_{\varepsilon}(x)=1-(1+\varepsilon)\cos x
\]
has exactly one zero on $[0,2\pi]$: 
\begin{equation}
x(\varepsilon)=\arccos\frac{1}{1+\varepsilon},\quad f'_{\varepsilon}(x(\varepsilon))=0.\label{x_epsilon}
\end{equation}
Choose $\varepsilon'>0$ so that $x(\varepsilon)\in[0, \pi / 6]$
for all $\varepsilon\in(0;\varepsilon']$. For this it is sufficient
that 
\[
\frac{1}{1+\varepsilon'}\geqslant\frac{\sqrt{3}}{2},\ \quad\varepsilon'\leqslant\frac{2}{\sqrt{3}}-1.
\]
Thus, we proved that $f_{\varepsilon}(x)$ has exactly one critical
point on $[0,\pi / 4]$, for any $\varepsilon\in(0,\varepsilon']$.
The graph of function $f_{\varepsilon}(x)$ for $\eps=0.1$ is as
follows: 
\[
 \hspace*{-20mm}\includegraphics[scale=0.25]{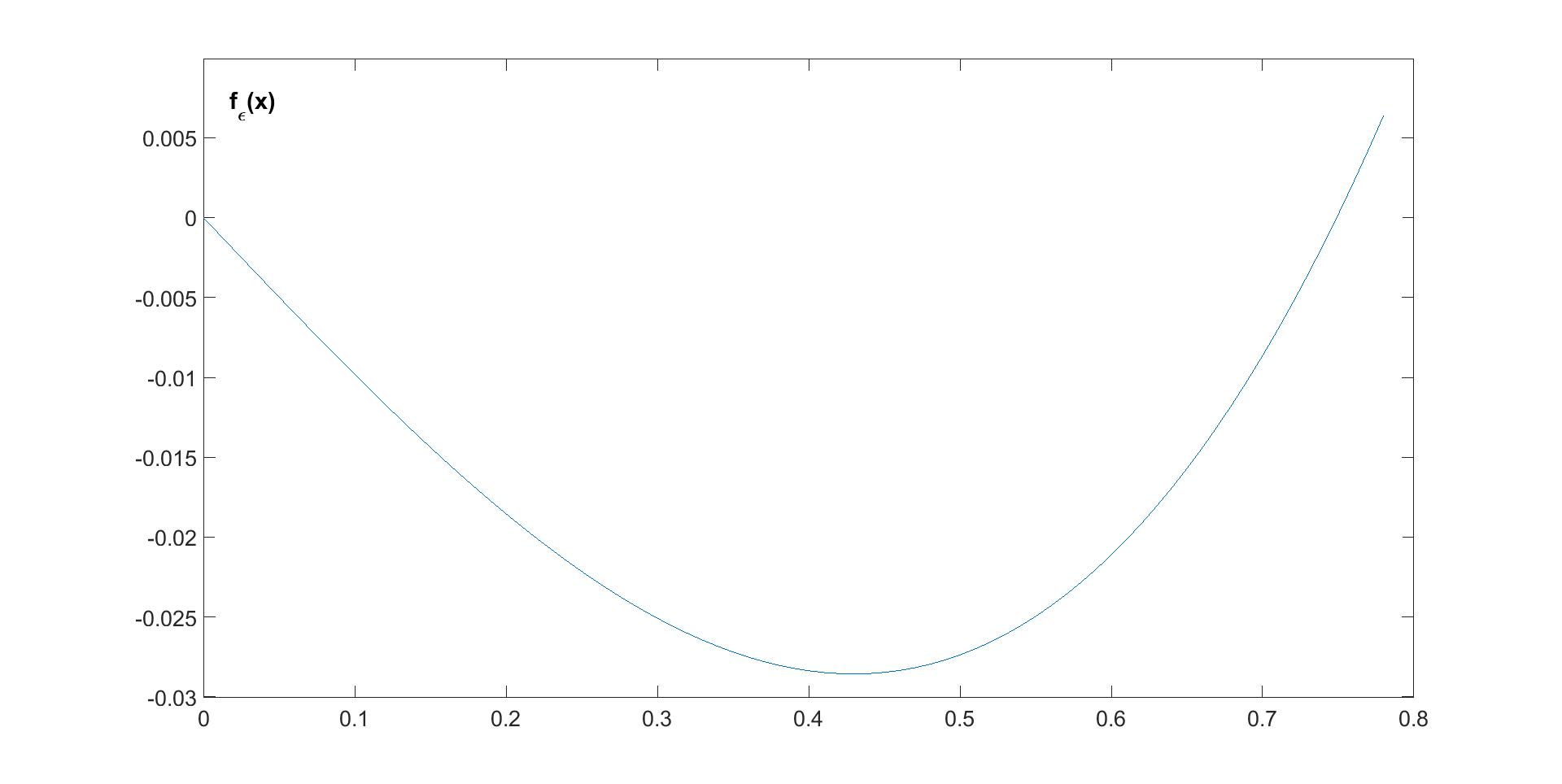}
\]

Now we write Taylor expansion of $f_{\varepsilon}(x)$ in a neighborhood
of $x(\varepsilon)$: 
\begin{equation}
f_{\varepsilon}(x)=f_{\varepsilon}(x(\varepsilon))+\frac{1}{2}f''_{\varepsilon}(x(\varepsilon))(x-x(\varepsilon))^{2}+\frac{1}{6}f'''_{\varepsilon}(x(\varepsilon))(x-x(\varepsilon))^{3}+R_{\varepsilon}(x),\label{epsbolzerfTailor}
\end{equation}
where 
\[
R_{\varepsilon}(x)=\frac{1}{3!}\int_{x(\varepsilon)}^{x}(x-s)^{3}f_{\varepsilon}^{(4)}(s)ds.
\]
Introduce the following notation 
\begin{align}
a_{\varepsilon}= & \ f_{\varepsilon}(x(\varepsilon))=x(\varepsilon)-(1+\varepsilon)\sin x(\varepsilon),\\
b_{\varepsilon}= & \ f''_{\varepsilon}(x(\varepsilon))=(1+\varepsilon)\sin x(\varepsilon).
\end{align}
As $f'''_{\varepsilon}(x)=(1+\varepsilon)\cos x$, by (\ref{x_epsilon})
we have $f'''_{\varepsilon}(x(\varepsilon))=1$. Note that $b_{\varepsilon}>0$.
Then (\ref{epsbolzerfTailor}) can be rewritten as follows:
\begin{equation}
f_{\varepsilon}(x)=a_{\varepsilon}+\frac{1}{2}b_{\varepsilon}(x-x(\varepsilon))^{2}+\frac{1}{6}(x-x(\varepsilon))^{3}+R_{\varepsilon}(x)=g_{\varepsilon}(x-x(\varepsilon))+R_{\varepsilon}(x),\label{epsbolzerfTailorViag}
\end{equation}
where 
\[
g_{\varepsilon}(x)=a_{\varepsilon}+\frac{1}{2}b_{\varepsilon}x^{2}+\frac{1}{6}x^{3}.
\]
The graph of the function $g_{\varepsilon}(x)$: 
\[
 \hspace*{-20mm}\includegraphics[scale=0.28]{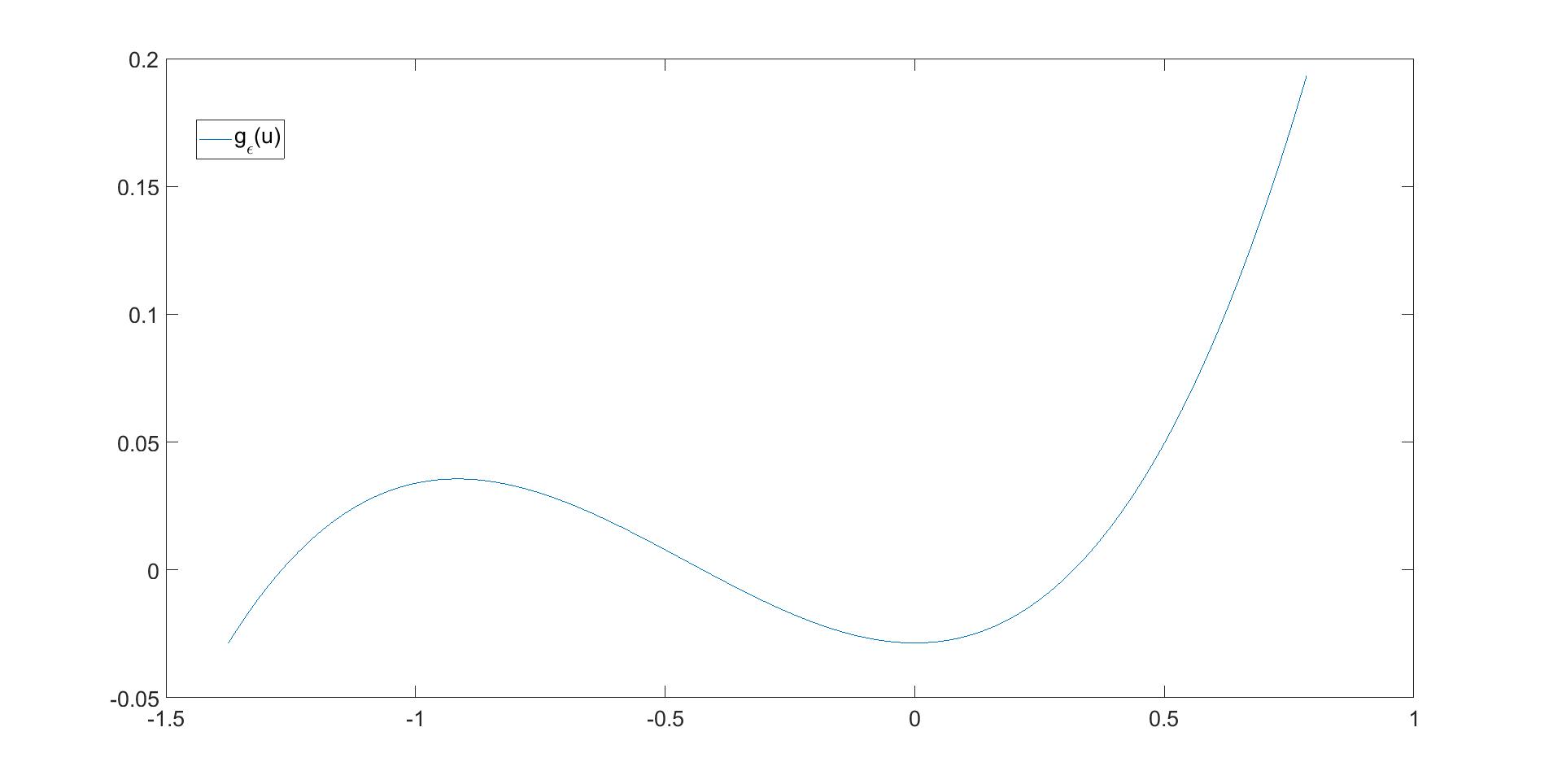}
\]

Note that for $\epsilon=0$ the function $g_{0}(x)$, introduced here,
coincides exactly wit the function $g_{0}(x)$, introduced during
the proof of Lemma \ref{gammalessonelemma}. This explains why we
used the same notation.

Now use Lemma \ref{phisubstitutebolzero} and change the variables
$x=\varphi_{\varepsilon}(u)$ in the integral  $L_{n}(\varepsilon)$.
We get 
\begin{align*}
  L_{n}(\eps)&=\int_{l_{\eps}}^{r_{\eps}}\frac{\sin(ng_{\eps}(u))}{\varphi_{\eps}(u)}\varphi'_{\eps}(u)\ du=\int_{l_{\eps}}^{r_{\eps}}\frac{\sin(ng_{\eps}(u))}{u-l_{\eps}}\frac{u-l_{\eps}}{\varphi_{\eps}(u)}\varphi'_{\eps}(u)\ du\\
  &=\int_{l_{\eps}}^{r_{\eps}}\frac{\sin(ng_{\eps}(u))}{u-l_{\eps}}h_{\eps}(u)\ du\\
&=h(l_{\eps})\int_{l_{\eps}}^{r_{\eps}}\frac{\sin(ng_{\eps}(u))}{u-l_{\eps}}\ du+\int_{l_{\eps}}^{r_{\eps}}\sin(ng_{\eps}(u))\frac{h_{\eps}(u)-h_{\eps}(l_{\eps})}{u-l_{\eps}}\ du,
\end{align*}
where 
\[
h_{\eps}(u)=\frac{u-l_{\eps}}{\varphi_{\eps}(u)}\varphi'_{\eps}(u)=\tilde{h}_{\eps}(u-l_{\eps}).
\]
The properties of the function $\varphi_{\eps}(u)$, described in Lemma
\ref{phisubstitutebolzero} and Lemma \ref{techLemmah}, show that
the second integral in the last formula is bounded uniformly in $n$
and $0<\eps\leqslant\eps'$. Thus, 
\[
L_{n}(\eps)=\tilde{L}_{n}(\eps)+O(1),
\]
where 
\[
\tilde{L}_{n}(\eps)=\int_{l_{\eps}}^{r_{\eps}}\frac{\sin(ng_{\eps}(u))}{u-l_{\eps}}\ du.
\]
Now we want to estimate the integral $\tilde{L}_{n}(\eps)$ by transforming
the integration interval to special curve in the complex plane. To
do this we need some transformations. Using theorem on major convergence
we have: 
\begin{equation}
\tilde{L}_{n}(\eps)=\lim_{\delta\rightarrow0+}\int_{l_{\eps}}^{r_{\eps}}\frac{\sin(ng_{\eps}(u))}{u-(l_{\eps}-\delta)}\ du.\label{201906070050}
\end{equation}
In fact, as $g(l_{\eps})=0$, the integrand is bounded uniformly in
$u\in[l_{\eps},r_{\eps}]$ for $\delta\geqslant0$: 
\[
\Bigl|\frac{\sin(ng_{\eps}(u))}{u-(l_{\eps}-\delta)}\Bigr|\leqslant\frac{n|g_{\eps}(u)|}{u-(l_{\eps}-\delta)}\leqslant n\max_{u\in[l_{\eps},r_{\eps}]}|g'_{\eps}(u)|\frac{u-l_{\eps}}{u-(l_{\eps}-\delta)}\leqslant n\max_{u\in[l_{\eps},r_{\eps}]}|g'_{\eps}(u)|.
\]
Thus it is possible to pass to the limit under the sign of the integral in (\ref{201906070050}).
 Now we can continue with (\ref{201906070050}): 
\[
\tilde{L}_{n}(\eps)=\lim_{\delta\rightarrow0+}\mathrm{Im}\biggl(\int_{l_{\eps}}^{r_{\eps}}\frac{\exp(ing_{\eps}(u))}{u-(l_{\eps}-\delta)}\ du\biggr)
=\lim_{\delta\rightarrow0+}\mathrm{Im}\biggl(\int_{s(\eps)}\frac{\exp(ing_{\eps}(z))}{z-(l_{\eps}-\delta)}\ dz\biggr),
\]
where the contour $s(\eps)$, shown in the picture, is the union of
three segments: $s_{1}(\eps)$ --- connecting the points $l_{\eps}$
and $A=(1 / \sqrt{2})l_{\eps}e^{-i\pi/4}$; $s_{2}(\eps)$
--- connecting the points $A$ and $B=(1/\sqrt{2})r_{\eps}e^{i\pi / 4}$; 
$s_{3}(\eps)$ --- connecting the points $B$ and $r_{\eps}$.

\begin{figure}[!hbtp]
   \hspace*{-20mm}\includegraphics[scale=0.28]{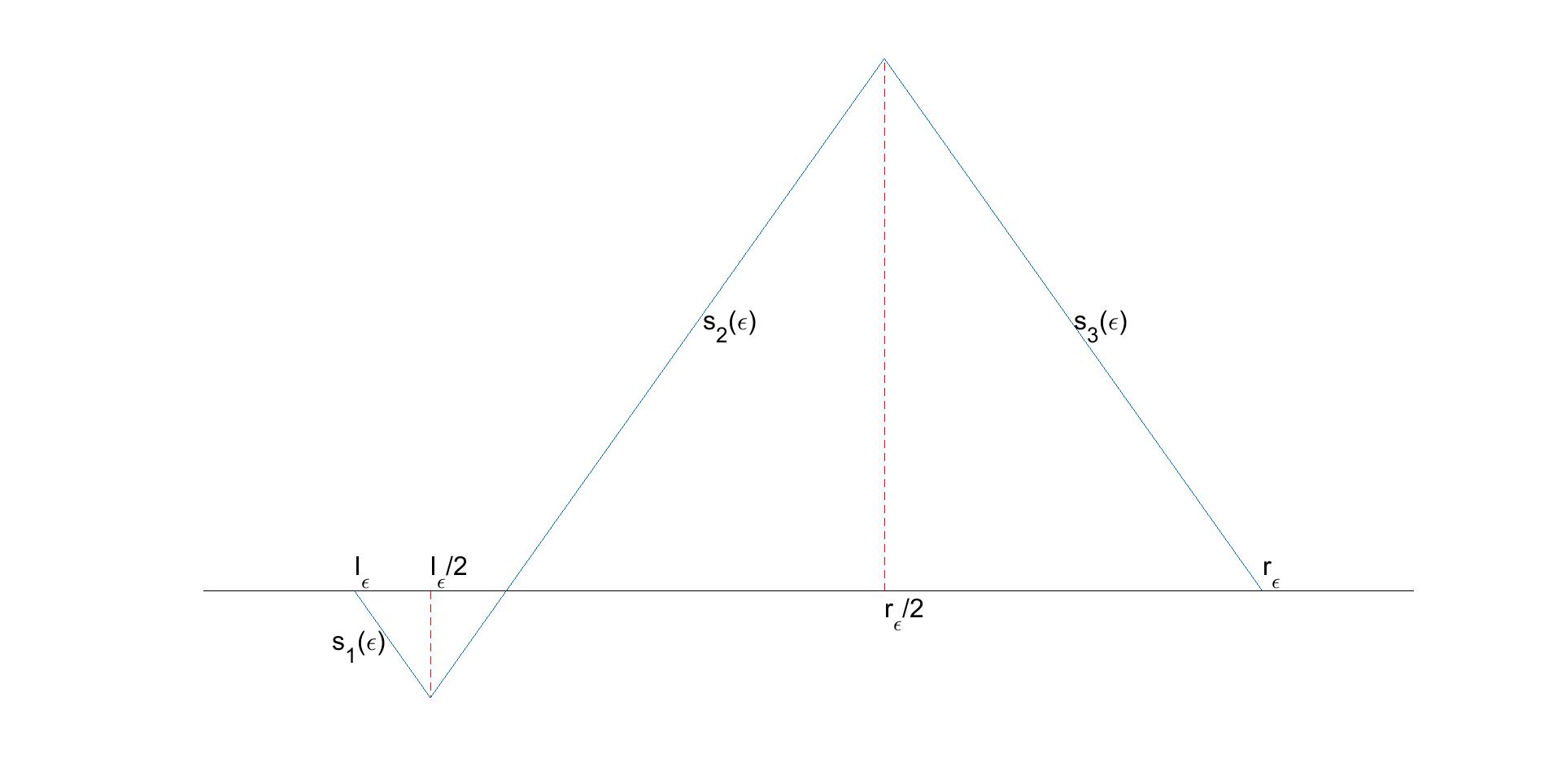}
\end{figure}
  
The change of the contour is possible due to analyticity of the integrand. Thus we can write 
\begin{equation}
\tilde{L}_{n}(\eps)=\lim_{\delta\rightarrow0+}\mathrm{Im}(I_{n}^{(1)}(\eps,\delta)+I_{n}^{(2)}(\eps,\delta)+I_{n}^{(3)}(\eps,\delta)) \label{201906082117}
\end{equation}
where 
\[
I_{n}^{(k)}(\eps,\delta)=\int_{s_{k}(\eps)}\frac{\exp(ing_{\eps}(z))}{z-(l_{\eps}-\delta)}\ dz,\ k=1,2,3.
\]
Now we should estimate all integrals $I_{n}^{(k)}(\eps,\delta)$ and
prove that the corresponding limits are bounded from above by absolute
constants, not depending on $n,\eps$. Note first that $b_{\eps},x(\eps),l_{\eps}$
are of the order $\sqrt{\eps}$ for small $\eps$, which makes more
complicated derivation of uniform estimates for first two integrals
$I_{n}^{(k)}(\eps,\delta),\ k=1,2$.

\subsubsection{Integral over $\boldsymbol{s_{1}(\eps)}$}

Consider first $I_{n}^{(1)}(\eps,\delta)$. The segment $s_{1}(\eps)$
can be presented as follows 
\[
z=l_{\eps}+\xi\rho,\quad \xi=e^{-i\pi / 4},\quad \rho\in[0,\tilde{\rho}],\quad \tilde{\rho}=\frac{|l_{\eps}|}{\sqrt{2}}.
\]
Then
\[
I_{n}^{(1)}(\eps,\delta)=\int_{s_{1}(\eps)}\frac{\exp(ing_{\eps}(z))}{z-(l_{\eps}-\delta)}\ dz=\xi\int_{0}^{\tilde{\rho}}\frac{\exp(ing_{\eps}(l_{\eps}+\xi\rho))}{\xi\rho+\delta}\ d\rho.
\]
  Using the Taylor expansion for $g_{\eps}(l_{\eps}+\xi\rho)$ 
 at  $l_{\eps}$, and taking into account that $g_{\eps}(l_{\eps})=0$, we get: 
\[
g_{\eps}(l_{\eps}+\xi\rho)=g'_{\eps}(l_{\eps})\xi\rho+\frac{g''_{\eps}(l_{\eps})}{2}\xi^{2}\rho^{2}+\frac{g_{\eps}^{(3)}(l_{\eps})}{6}\xi^{3}\rho^{3}=u_{1}(\rho)+iv_{1}(\rho),
\]
where 
\[
u_{1}(\rho)=\mathrm{Re}(g_{\eps}(l_{\eps}+\xi\rho))=\kappa g'_{\eps}(l_{\eps})\rho-\frac{\kappa}{6}\rho^{3},
\]
\[
v_{1}(\rho)=\mathrm{Im}(g_{\eps}(l_{\eps}+\xi\rho))=-\kappa g'_{\eps}(l_{\eps})\rho-\frac{g''_{\eps}(l_{\eps})}{2}\rho^{2}-\frac{\kappa}{6}\rho^{3},
\]
where $\kappa=1 / \sqrt{2}$. Then 
\begin{align*}
  &\mathrm{Im}(I_{n}^{(1)}(\eps,\delta))\\
  &\quad =\mathrm{Im}\Biggl(\xi\int_{0}^{\tilde{\rho}}\frac{\exp(ing_{\eps}(l_{\eps}+\xi\rho))}{\xi\rho+\delta}\ d\rho\Biggr)\\
  &\quad =\mathrm{Im}\Biggl(\xi\int_{0}^{\tilde{\rho}}\frac{\exp(inu_{1}(\rho))}{\xi\rho+\delta}e^{-nv_{1}(\rho)}\ d\rho\Biggr)
\\
                            &\quad =\int_{0}^{\tilde{\rho}}\mathrm{Im}\left(\xi e^{inu_{1}(\rho)}(\bar{\xi}\rho+\delta)\right)\frac{e^{-nv_{1}(\rho)}}{|\xi\rho+\delta|^{2}}\ d\rho\\
  &\quad =\int_{0}^{\tilde{\rho}}\mathrm{Im}\left(e^{inu_{1}(\rho)}(\rho+\xi\delta)\right)\frac{e^{-nv_{1}(\rho)}}{|\rho+\bar{\xi}\delta|^{2}}\ d\rho
\\
&\quad =\int_{0}^{\tilde{\rho}}(\cos(nu_{1}(\rho))(-\kappa\delta)+\sin(nu_{1}(\rho))(\rho+\kappa\delta))\frac{e^{-nv_{1}(\rho)}}{(\rho+\kappa\delta)^{2}+(\kappa\delta)^{2}}\ d\rho
\\
&\quad =Z_{n}^{(1)}(\eps,\delta)+Z_{n}^{(2)}(\eps,\delta),
\end{align*}
where 
\[
Z_{n}^{(1)}(\eps,\delta)=\int_{0}^{\tilde{\rho}}(\cos(nu_{1}(\rho))(-\kappa\delta)+\sin(nu_{1}(\rho))\kappa\delta)\frac{e^{-nv_{1}(\rho)}}{(\rho+\kappa\delta)^{2}+(\kappa\delta)^{2}}\ d\rho,
\]
\[
Z_{n}^{(2)}(\eps,\delta)=\int_{0}^{\tilde{\rho}}\rho\sin(nu_{1}(\rho))\frac{e^{-nv_{1}(\rho)}}{(\rho+\kappa\delta)^{2}+(\kappa\delta)^{2}}\ d\rho .
\]
One cannot take the limit $\delta\to+0$ inside the integral, as there
is no uniform convergence and the integrand equals $c/\delta$ at
the point $\rho=0$, i.e.\ it is not uniformly bounded. Now we will
prove that for all $\delta\geqslant0$, the absolute value of the integral
$\mathrm{Im}(I_{n}^{(1)}(\eps,\delta))$ can be estimated from above
by a value not depending on $n,\eps,\delta$. We should study
first the function $v_{1}(\rho)$. Let us show first that $v_{1}(\rho)$
is convex for $\rho\geqslant0$. Indeed, we have: 
\[
v_{1}''(\rho)=-g''_{\eps}(l_{\eps})-\kappa\rho.
\]
Note, using the inequality (\ref{bepslepsxepsabs}), that 
\[
g''_{\eps}(l_{\eps})=b_{\eps}+l_{\eps}>b_{\eps}-x({\eps})=-a_{\eps}>0.
\]
Hence
\[
v_{1}''(\rho)<0\quad \mbox{for}\ \rho\geqslant0.
\]
Upward convexity of $v_{1}(\rho)$ follows. Since $v_{1}(0)=0$, 
by upward convexity we have the following inequality for all $\rho\in[0,\tilde{\rho}]$:
\begin{equation}
v_{1}(\rho)\geqslant\rho\frac{v_{1}(\tilde{\rho})}{\tilde{\rho}}.\label{v1rhoineq}
\end{equation}
Let us find now the value of $v_{1}(\tilde{\rho})$. To do this we
shall write down the derivatives of $g_{\eps}$ at point $l_{\eps}$
through $a_{\eps},b_{\eps}$: 
\[
g'_{\eps}(l_{\eps})=b_{\eps}l_{\eps}+\frac{l_{\eps}^{2}}{2},\quad g''_{\eps}(l_{\eps})=b_{\eps}+l_{\eps}.
\]
Using $l_{\eps}<0$ and (\ref{bepslepsxepsabs}) we have: 
\begin{align*}
  v_{1}(\tilde{\rho})&=-\kappa^{2}\Bigl(b_{\eps}l_{\eps}+\frac{l_{\eps}^{2}}{2}\Bigr) |l_{\eps}|-\frac{\kappa^{2}}{2}(b_{\eps}+l_{\eps})l_{\eps}^{2}-\frac{\kappa^{4}}{6}|l_{\eps}|^{3}\\
  &=\kappa^{2}l_{\eps}^{2}\Bigl(b_{\eps}+\frac{1}{2}l_{\eps}-\frac{1}{2}(b_{\eps}+l_{\eps})+\frac{1}{12}l_{\eps}\Bigr)=
\\
&=\frac{\kappa^{2}l_{\eps}^{2}}{2}\Bigl(b_{\eps}+\frac{1}{6}l_{\eps}\Bigr)>\frac{\kappa^{2}l_{\eps}^{2}}{2}\left(b_{\eps}+l_{\eps}\right)>0.
\end{align*}
Thus, from (\ref{v1rhoineq}) it follows, in particular,
that $v(\rho)\geqslant0$ if $\rho\in[0,\tilde{\rho}]$.

Now we want to estimate $Z_{n}^{(1)}(\eps,\delta)$: 
\begin{align*}
  |Z_{n}^{(1)}(\eps,\delta)| &\leqslant2\delta\int_{0}^{\tilde{\rho}}\frac{1}{(\rho+\kappa\delta)^{2}+(\kappa\delta)^{2}}\ d\rho
                               \qquad \mbox{[substitute \ensuremath{\rho=\delta u}]} \\
  &=2\int_{0}^{\tilde{\rho}/\delta}\frac{1}{(u+\kappa)^{2}+\kappa^{2}}\ du
\leqslant2\int_{0}^{+\infty}\frac{1}{u^{2}+\kappa^{2}}\ du\leqslant4\frac{\pi}{2}=2\pi.
\end{align*}
Note first that, as $u_{1}(0)=0$, by Lagrange theorem $$|u_{1}(\rho)|\leqslant\rho\max_{\rho\in[0,\tilde{\rho}]}|u'_{1}(\rho)|=\rho M$$
for all $\rho\in[0,\tilde{\rho}]$, where $M=\max_{\rho\in[0,\tilde{\rho}]}|u'_{1}(\rho)|$.
Using (\ref{v1rhoineq}) we get the bounds: 
\begin{align*}
  |Z_{n}^{(2)}(\eps,\delta)| &\leqslant nM\int_{0}^{\tilde{\rho}}\rho^{2}\frac{e^{-nv_{1}(\rho)}}{(\rho+\kappa\delta)^{2}+(\kappa\delta)^{2}}\ d\rho
                               \leqslant nM\int_{0}^{\tilde{\rho}}e^{-nv_{1}(\rho)}d\rho \\
&  \leqslant nM\int_{0}^{\tilde{\rho}}e^{-n\rho v_{1}(\tilde{\rho}) / \tilde{\rho}}d\rho
=nM\frac{1}{n v_{1}(\tilde{\rho}) / \tilde{\rho}}\left(1-e^{-nv_{1}(\tilde{\rho})}\right)\leqslant M\frac{\tilde{\rho}}{v_{1}(\tilde{\rho})}.
\end{align*}
Estimate now $M=\max_{\rho\in[0,\tilde{\rho}]}|u'_{1}(\rho)|$. For
$\rho\in[0,\tilde{\rho}]$ we have: 
\begin{align*}
  |u'_{1}(\rho)|&=\Bigl|\kappa g'_{\eps}(l_{\eps})-\frac{\kappa}{2}\rho^{2}\Bigr|
  \leqslant\kappa|g'_{\eps}(l_{\eps})|+\frac{\kappa^{2}}{2}l_{\eps}^{2}\\
&  \leqslant\kappa\Bigl(b_{\eps}|l_{\eps}|+\frac{l_{\eps}^{2}}{2}\Bigr)+\frac{\kappa^{2}}{2}l_{\eps}^{2}
  \leqslant|l_{\eps}|\left(b_{\eps}+2|l_{\eps}|\right).
\end{align*}
Continuing estimation of $|Z_{n}^{(2)}(\eps,\delta)|$, using (\ref{bepslepsxepsrel}) we get
\begin{align*}
  |Z_{n}^{(2)}(\eps,\delta)| &\leqslant|l_{\eps}|\left(b_{\eps}+2|l_{\eps}|\right)\frac{\kappa|l_{\eps}|}{(\kappa^{2}l_{\eps}^{2} / 2) (b_{\eps}+ l_{\eps} / 6 )}
                               \leqslant4\frac{b_{\eps}+2|l_{\eps}|}{b_{\eps}+ l_{\eps} / 6}\\
  &\leqslant4\frac{b_{\eps} / |l_{\eps}| +2}{b_{\eps} / |l_{\eps}|- (1/6)}\leqslant4\frac{6}{1- 1/6}=\frac{144}{5}.
\end{align*}
Thus we have shown that for all $\delta>0$ and all $\eps\in(0,\eps'],n\geqslant1$:
\begin{equation}
|\mathrm{Im}(I_{n}^{(1)}(\eps,\delta))|\leqslant2\pi+\frac{144}{5}\leqslant36.\label{201906082115}
\end{equation}

\subsubsection{Integral over $\boldsymbol{s_{2}(\eps)}$}
The segment $s_{2}(\eps)$ can be expressed as 
\[
z=\eta\rho,\quad \eta=e^{i \pi / 4},\quad \rho\in[-\tilde{\rho},\hat{\rho}],\quad \tilde{\rho}=\frac{|l_{\eps}|}{\sqrt{2}},\quad \hat{\rho}=\frac{r_{\eps}}{\sqrt{2}},
\]
hence we have 
\[
I_{n}^{(2)}(\eps,\delta)=\int_{s_{2}(\eps)}\frac{\exp(ing_{\eps}(z))}{z-(l_{\eps}-\delta)}\ dz=\eta\int_{-\tilde{\rho}}^{\hat{\rho}}\frac{\exp(ing_{\eps}(\eta\rho))}{\eta\rho-(l_{\eps}-\delta)}\ d\rho.
\]
Note that for all $\rho\in[-\tilde{\rho},\hat{\rho}]$ and $\delta\geqslant0$:
\[
\Bigl|\frac{\exp(ing_{\eps}(\eta\rho))}{\eta\rho-(l_{\eps}-\delta)}\Bigr|\leqslant\frac{c}{|l_{\eps}|}
\]
for some constant $c>0$ not depending on $\rho$ and $\delta$.
Then by major convergence theorem we have 
\[
\lim_{\delta\rightarrow0+}\mathrm{Im}(I_{n}^{(2)}(\eps,\delta))=\mathrm{Im}\left(\eta\int_{-\tilde{\rho}}^{\hat{\rho}}\frac{\exp(ing_{\eps}(\eta\rho))}{\eta\rho-l_{\eps}}\ d\rho\right)=:W_{n}(\eps).
\]
Now let us find $g_{\eps}(\eta\rho)$: 
\[
g_{\eps}(\eta\rho)=a_{\eps}+\frac{b_{\eps}}{2}\eta^{2}\rho^{2}+\frac{1}{6}\eta^{3}\rho^{3}=u_{2}(\rho)+iv_{2}(\rho),
\]
where 
\[
u_{2}(\rho)=\mathrm{Re}(g_{\eps}(\eta\rho))=a_{\eps}-\frac{\kappa}{6}\rho^{3},
\]
\[
v_{2}(\rho)=\mathrm{Im}(g_{\eps}(\eta\rho))=\frac{b_{\eps}}{2}\rho^{2}+\frac{\kappa}{6}\rho^{3},
\]
and, as above, we used notation $\kappa=1 / \sqrt{2}$.

Let us come back to the equality for $W_{n}(\eps)$: 
\begin{align*}
  W_{n}(\eps)&=\int_{-\tilde{\rho}}^{\hat{\rho}}\mathrm{Im}\Bigl(\eta\frac{e^{inu_{2}(\rho)}}{\eta\rho-l_{\eps}}\Bigr)e^{-nv_{2}(\rho)}d\rho\\
  &=\int_{-\tilde{\rho}}^{\hat{\rho}}\mathrm{Im}\left(\eta e^{inu_{2}(\rho)}(\bar{\eta}\rho-l_{\eps})\right)\frac{e^{-nv_{2}(\rho)}}{|\eta\rho-l_{\eps}|^{2}}d\rho
\\
&=\int_{-\tilde{\rho}}^{\hat{\rho}}\mathrm{Im}\left(e^{inu_{2}(\rho)}(\rho-l_{\eps}\eta)\right)\frac{e^{-nv_{2}(\rho)}}{|\rho-\bar{\eta}l_{\eps}|^{2}}d\rho
\\
&=\int_{-\tilde{\rho}}^{\hat{\rho}}\frac{\cos(nu_{2}(\rho))(-l_{\eps}\kappa)+\sin(nu_{2}(\rho))(\rho-l_{\eps}\kappa)}{(\rho-\kappa l_{\eps})^{2}+(\kappa l_{\eps})^{2}}e^{-nv_{2}(\rho)}d\rho
\\
             &=\int_{-\tilde{\rho}}^{\hat{\rho}}\frac{\cos(nu_{2}(\rho))\tilde{\rho}+\sin(nu_{2}(\rho))(\rho+\tilde{\rho})}{(\rho+\tilde{\rho})^{2}+\tilde{\rho}^{2}}e^{-nv_{2}(\rho)}d\rho\\
  &=W_{n}^{(1)}(\eps)+W_{n}^{(2)}(\eps),
\end{align*}
where 
\[
W_{n}^{(1)}(\eps)=\tilde{\rho}\int_{-\tilde{\rho}}^{\hat{\rho}}\frac{\cos(nu_{2}(\rho))+\sin(nu_{2}(\rho))}{(\rho+\tilde{\rho})^{2}+\tilde{\rho}^{2}}e^{-nv_{2}(\rho)}d\rho,
\]
\[
W_{n}^{(2)}(\eps)=\int_{-\tilde{\rho}}^{\hat{\rho}}\rho\frac{\sin(nu_{2}(\rho))}{(\rho+\tilde{\rho})^{2}+\tilde{\rho}^{2}}e^{-nv_{2}(\rho)}d\rho.
\]
Now we estimate $W_{n}^{(1)}(\eps),W_{n}^{(2)}(\eps)$. Note first
that the function $v_{2}(\rho)$ is non-negative for $\rho\in[-\tilde{\rho},\hat{\rho}]$:
\begin{equation}
  v_{2}(\rho)=\frac{\rho^{2}}{2}\Bigl(b_{\eps}+\frac{\kappa}{3}\rho\Bigr)\geqslant\frac{\rho^{2}}{2}\Bigl(b_{\eps}+\frac{\kappa^{2}}{3}l_{\eps}\Bigr)
  =\frac{\rho^{2}}{2}\Bigl(b_{\eps}+\frac{1}{6}l_{\eps}\Bigr)>0.\label{v2ineqminor}
\end{equation}
The last inequality follows from (\ref{bepslepsxepsabs}). Then 
\begin{align*}
  |W_{n}^{(1)}(\eps)|&\leqslant2\tilde{\rho}\int_{-\tilde{\rho}}^{\hat{\rho}}\frac{1}{(\rho+\tilde{\rho})^{2}+\tilde{\rho}^{2}}d\rho\leqslant2\tilde{\rho}\int_{-\tilde{\rho}}^{+\infty}\frac{1}{(\rho+\tilde{\rho})^{2}+\tilde{\rho}^{2}}d\rho\\
  &\mbox{[substitute \ensuremath{\rho=\tilde{\rho}u}]} \qquad
    =2\int_{-1}^{+\infty}\frac{1}{(u+1)^{2}+1}du\\
  &=2\int_{-1}^{0}\frac{1}{(u+1)^{2}+1}du+2\int_{0}^{+\infty}\frac{1}{(u+1)^{2}+1}du\leqslant2+\pi.
\end{align*}
Now we estimate $W_{n}^{(2)}(\eps)$. To do this we rewrite it as
follows 
\[
W_{n}^{(2)}(\eps)=W_{n}^{(2),-}(\eps)+W_{n}^{(2),+}(\eps),
\]
where 
\[
W_{n}^{(2),-}(\eps)=\int_{-\tilde{\rho}}^{0}\rho\frac{\sin(nu_{2}(\rho))}{(\rho+\tilde{\rho})^{2}+\tilde{\rho}^{2}}\, e^{-nv_{2}(\rho)}d\rho,
\]
\[
W_{n}^{(2),+}(\eps)=\int_{0}^{\hat{\rho}}\rho\frac{\sin(nu_{2}(\rho))}{(\rho+\tilde{\rho})^{2}+\tilde{\rho}^{2}} \, e^{-nv_{2}(\rho)}d\rho.
\]
The integral $W_{n}^{(2),-}(\eps)$ is estimated similarly to $W_{n}^{(1)}(\eps)$:
\[
|W_{n}^{(2),-}(\eps)|\leqslant\int_{-\tilde{\rho}}^{0}\frac{|\rho|}{(\rho+\tilde{\rho})^{2}+\tilde{\rho}^{2}}d\rho=\int_{-1}^{0}\frac{|u|}{(u+1)^{2}+1}du\leqslant\frac{1}{2}.
\]
While analyzing $W_{n}^{(2),+}(\eps)$, a trivial constant estimate
for the sine does not allow to obtain an estimate of the corresponding
integral uniform in $n,\eps$. One of the reasons is that the integrand
(not taking the sine into account) for small $\eps$ at zero has order
$1 / \tilde{\rho} \sim 1 / \sqrt{\eps}$. Write now $W_{n}^{(2),+}(\eps)$
as follows: 
\begin{align*}
  W_{n}^{(2),+}(\eps)&=\sin(na_{\eps})\int_{0}^{\hat{\rho}}\rho\frac{\cos\left( \kappa n\rho^{3} / 6 \right)}{(\rho+\tilde{\rho})^{2}+\tilde{\rho}^{2}} \,
                       e^{-nv_{2}(\rho)}d\rho\\
  &\quad {} -\cos(na_{\eps})\int_{0}^{\hat{\rho}}\rho\frac{\sin\left( \kappa n\rho^{3} / 6 \right)}{(\rho+\tilde{\rho})^{2}+\tilde{\rho}^{2}} \, e^{-nv_{2}(\rho)}d\rho
\\
&=\sin(na_{\eps})C_{n}(\eps)-\cos(na_{\eps})S_{n}(\eps),
\end{align*}
where 
\[
  C_{n}(\eps)=\int_{0}^{\hat{\rho}}\rho\frac{\cos\left( \kappa n\rho^{3} / 6 \right)}{(\rho+\tilde{\rho})^{2}+\tilde{\rho}^{2}}e^{-nv_{2}(\rho)}d\rho,\quad
  S_{n}(\eps)=\int_{0}^{\hat{\rho}}\rho\frac{\sin\left( \kappa n\rho^{3} / 6 \right)}{(\rho+\tilde{\rho})^{2}+\tilde{\rho}^{2}}e^{-nv_{2}(\rho)}d\rho
\]
and we have to estimate it. Put $d_{\eps}= \left(b_{\eps}+ l_{\eps} / 6 \right) /2$.
Using  (\ref{v2ineqminor}) we get: 
\begin{align*}
|C_{n}(\eps)| &\leqslant\int_{0}^{\hat{\rho}}\frac{\rho}{(\rho+\tilde{\rho})^{2}+\tilde{\rho}^{2}}\, e^{-nv_{2}(\rho)}d\rho\leqslant\int_{0}^{\hat{\rho}}\frac{\rho}{(\rho+\tilde{\rho})^{2}+\tilde{\rho}^{2}}\, e^{-nd_{\eps}\rho^{2}}d\rho
\\
              &\leqslant\int_{0}^{\infty}\frac{\rho}{(\rho+\tilde{\rho})^{2}+\tilde{\rho}^{2}}\, e^{-nd_{\eps}\rho^{2}}d\rho
                \qquad \mbox{[substitute \ensuremath{\rho= u / \sqrt{nd_{\eps}}}]}
\\
              &=\int_{0}^{\infty}\frac{u}{(u+\tilde{\rho}\sqrt{nd_{\eps}})^{2}+(\tilde{\rho}\sqrt{nd_{\eps}})^{2}} \, e^{-u^{2}}du\\
  &\leqslant\frac{1}{(\tilde{\rho}\sqrt{nd_{\eps}})^{2}}\int_{0}^{\infty}ue^{-u^{2}}du=\frac{1}{2n\tilde{\rho}^{2}d_{\eps}}.
\end{align*}
Now we estimate $S_{n}(\eps)$. Note that $v_{2}(\rho)> \kappa \rho^{3} / 6$
for $\rho\geqslant0$. Consequently, 
\begin{align*}
  |S_{n}(\eps)|&\leqslant n\frac{\kappa}{6}\int_{0}^{\hat{\rho}}\frac{\rho^{4}}{(\rho+\tilde{\rho})^{2}+\tilde{\rho}^{2}} \, e^{-nv_{2}(\rho)}d\rho \\
               &\leqslant n\frac{\kappa}{6}\int_{0}^{\hat{\rho}}\rho^{2}e^{-n \kappa \rho^{3} / 6} d\rho\leqslant n\frac{\kappa}{6}
                 \int_{0}^{\infty}\rho^{2}e^{-n \kappa \rho^{3} / 6 }d\rho=\frac{1}{3}.
\end{align*}
Then 
\[
|W_{n}^{(2),+}(\eps)|\leqslant\frac{|\sin(na_{\eps})|}{2n\tilde{\rho}^{2}d_{\eps}}+\frac{1}{3}=\frac{|a_{\eps}|}{2\tilde{\rho}^{2}d_{\eps}}\frac{|\sin(na_{\eps})|}{n|a_{\eps}|}+\frac{1}{3}\leqslant\frac{|a_{\eps}|}{2\tilde{\rho}^{2}d_{\eps}}+\frac{1}{3}.
\]
The first term here should be estimated separately. Note that by definition
of number $l_{\eps}$ we have: 
\[
a_{\eps}+\frac{b_{\eps}}{2}l_{\eps}^{2}+\frac{1}{6}l_{\eps}^{3}=0.
\]
Then $a_{\eps}=- b_{\eps} l_{\eps}^{2} / 2 - l_{\eps}^{3} / 6 $.
From this, using inequality (\ref{bepslepsxepsrel}), we get: 
\[
  \frac{|a_{\eps}|}{2\tilde{\rho}^{2}d_{\eps}}
  = \frac{l_{\eps}^{2}}{2}
  \frac{|b_{\eps}+ l_{\eps}/3|}{\kappa^{2}l_{\eps}^{2}\left(b_{\eps}+ l_{\eps} / 6 \right)}
  =\frac{b_{\eps}+ l_{\eps} / 3}{b_{\eps}+ l_{\eps} / 6}
  =\frac{b_{\eps} / |l_{\eps}| - 1/3}{b_{\eps} / |l_{\eps}|- 1/6}
  \leqslant\frac{4-1/3}{1-1/6}=\frac{22}{5}
\]
and thus 
\[
|W_{n}^{(2),+}(\eps)|\leqslant\frac{22}{5}+\frac{1}{3}\leqslant5.
\]
Finally we get: 
\begin{equation}
\bigl| \lim_{\delta\rightarrow0+}\mathrm{Im}(I_{n}^{(2)}(\eps,\delta))\bigr| =|W_{n}(\eps)|\leqslant2+\pi+\frac{1}{2}+5\leqslant11.\label{201906082114}
\end{equation}

\subsubsection{Integral over $\boldsymbol{s_{3}(\eps)}$}

We want to estimate the integral over the interval $I_{n}^{(3)}(\eps,\delta)$.
The interval $s_{3}(\eps)$ can be written as follows: 
\[
z=\eta\hat{\rho}+\bar{\eta}\rho,\quad \eta=e^{i \pi / 4},\quad \rho\in[0,\hat{\rho}],\quad \hat{\rho}=\frac{r_{\eps}}{\sqrt{2}}.
\]
This gives 
\[
I_{n}^{(3)}(\eps,\delta)=\int_{s_{3}(\eps)}\frac{\exp(ing_{\eps}(z))}{z-(l_{\eps}-\delta)}\ dz=\bar{\eta}\int_{0}^{\hat{\rho}}\frac{\exp(ing_{\eps}(\eta\hat{\rho}+\bar{\eta}\rho))}{\eta\hat{\rho}+\bar{\eta}\rho-(l_{\eps}-\delta)}\ d\rho.
\]
The same arguments as for $I_{n}^{(2)}(\eps,\delta)$ show that it  is possible to pass 
to the limit in the integrand: 
\[
\lim_{\delta\rightarrow0+}\mathrm{Im}(I_{n}^{(3)}(\eps,\delta))=\mathrm{Im}\left(\bar{\eta}\int_{0}^{\hat{\rho}}\frac{\exp(ing_{\eps}(\eta\hat{\rho}+\bar{\eta}\rho))}{\eta\hat{\rho}+\bar{\eta}\rho-l_{\eps}}\ d\rho\right)=:X_{n}(\eps).
\]
Now let us find $g_{\eps}(\eta\hat{\rho}+\bar{\eta}\rho))$: 
\begin{align*}
  g_{\eps}(\eta\hat{\rho}+\bar{\eta}\rho))&=a_{\eps}+\frac{b_{\eps}}{2}(\eta\hat{\rho}+\bar{\eta}\rho)^{2}+\frac{1}{6}(\eta\hat{\rho}+\bar{\eta}\rho)^{3}\\
  &=a_{\eps}+\frac{b_{\eps}}{2}(i\hat{\rho}^{2}-i\rho^{2}+2\hat{\rho}\rho)
\\
                                          &\quad {} +\frac{1}{6}(\eta^{3}\hat{\rho}^{3}+3\eta\hat{\rho}^{2}\rho+3\bar{\eta}\hat{\rho}\rho^{2}+\bar{\eta}^{3}\rho^{3})\\
  &=u_{3}(\rho)+iv_{3}(\rho),
\end{align*}
where 
\begin{align*}
u_{3}(\rho)&=\mathrm{Re}(g_{\eps}(\eta\hat{\rho}+\bar{\eta}\rho)),
\\
v_{3}(\rho)&=\mathrm{Im}(g_{\eps}(\eta\hat{\rho}+\bar{\eta}\rho))=\frac{b_{\eps}}{2}(\hat{\rho}^{2}-\rho^{2})+\frac{\kappa}{6}(\hat{\rho}^{3}+3\hat{\rho}^{2}\rho-3\hat{\rho}\rho^{2}-\rho^{3})
\\
&=(\hat{\rho}-\rho)\Bigl(\frac{b_{\eps}}{2}(\hat{\rho}+\rho)+\frac{\kappa}{6}(\hat{\rho}^{2}+\rho^{2}+4\hat{\rho}\rho)\Bigr)\geqslant0
\end{align*}
for $\rho\in[0,\hat{\rho}]$. Then 
\[
X_{n}(\eps)=\int_{0}^{\hat{\rho}}e^{-nv_{3}(\rho)}\mathrm{Im}\left(\bar{\eta}\frac{e^{inu_{3}(\rho)}}{\eta\hat{\rho}+\bar{\eta}\rho-l_{\eps}}\right)\ d\rho.
\]
Note that for all $\rho\in[0,\hat{\rho}]$ the following inequality
holds: 
\begin{equation}
\Bigl| \frac{1}{\eta\hat{\rho}+\bar{\eta}\rho-l_{\eps}}\Bigr| \leqslant\frac{\sqrt{2}}{r_{\eps}+|l_{\eps}|}.\label{201906082106}
\end{equation}
Indeed, we shall drop a perpendicular from the point $l_{\eps}$
on the straight line containing the interval $s_{3}(\eps)$ and the
point $r_{\eps}$. It is easy to see that the length of this perpendicular
equals $(r_{\eps}+|l_{\eps}|) / \sqrt{2}$ and is less than the
distance from the point $l_{\eps}$ to the segment $s_{3}(\eps)$. From
this the inequality (\ref{201906082106}) follows. And also for $X_{n}(\eps)$
we get the inequality: 
\begin{equation}
|X_{n}(\eps)|\leqslant\frac{\sqrt{2}\hat{\rho}}{r_{\eps}+|l_{\eps}|}=\frac{r_{\eps}}{r_{\eps}+|l_{\eps}|}\leqslant1.\label{201906082116}
\end{equation}

Finally, from equality (\ref{201906082117}) and inequalities (\ref{201906082115}), (\ref{201906082114}), (\ref{201906082116})
we get the following bound for $\tilde{L}_{n}(\eps)$: 
\[
|\tilde{L}_{n}(\eps)|\leqslant1+11+36\leqslant49.
\]
Thus, Lemma \ref{gammagreateronelemma} is completely proved.
\end{proof}

\subsubsection{Lemma on the change of variables}

One of the key statements necessary for the proof of uniform boundedness
of $L_{n}(\eps)$ in case $\eps>0$ is the following lemma on the
change of variables.
\begin{lemma}[On the change of variables] \label{phisubstitutebolzero} 
There exists $0<\varepsilon'<1$ such that for all $\varepsilon\in(0;\varepsilon']$ there are $l_{\varepsilon}<0<r_{\varepsilon}$ so that the
following statements hold:
\begin{enumerate}
\item There exists continuous increasing function 
\[
\varphi_{\varepsilon}:[l_{\varepsilon},r_{\varepsilon}]\to[0,\frac{\pi}{4}],\quad\varphi_{\varepsilon}(l_{\varepsilon})=0,\ \varphi_{\varepsilon}(0)=x(\varepsilon),\ \varphi_{\varepsilon}(r_{\varepsilon})=\frac{\pi}{4}
\]
such that
\[
f_{\varepsilon}(\varphi_{\varepsilon}(u))=g_{\varepsilon}(u).
\]
\item $\varphi_{\varepsilon}(u)\in C^{1}([l_{\varepsilon},r_{\varepsilon}])$,
moreover, for any $u\in[l_{\varepsilon},r_{\varepsilon}]$ there exists $\varphi''_{\varepsilon}(u)$ 
and  the following inequalities hold: 
\[
0<c_{1}\leqslant\varphi'_{\varepsilon}(u)\leqslant c_{2},
\]
\[
|\varphi''_{\varepsilon}(u)|\leqslant c_{3},
\]
\begin{equation}
c_{4}\leqslant r_{\varepsilon}\leqslant c_{5}  \label{repsineq}
\end{equation}
for some positive constants $c_{1},c_{2},c_{3},c_{4},c_{5},c_{6}$
not depending on $\varepsilon$.
\item Also the following inequalities hold: 
\begin{equation}
-b_{\eps}<-x(\eps)<l_{\eps}<0,\label{bepslepsxepsabs}
\end{equation}
\begin{equation}
\frac{b_{\eps}}{|l_{\eps}|}<4.\label{bepslepsxepsrel}
\end{equation}
\end{enumerate}
\end{lemma}
\begin{proof}
  Let us prove the first assertion. Consider two cases.

  1. $x\in[x(\varepsilon); \pi / 4]$. In this case, the function $f_{\varepsilon}(x)$
is monotone increasing. For the derivative of $g_{\varepsilon}(u)$
we use: 
\begin{equation}
g'_{\varepsilon}(u)=b_{\varepsilon}u+\frac{1}{2}u^{2}=u\Bigl(b_{\varepsilon}+\frac{1}{2}u\Bigr).\label{gder}
\end{equation}
Then, for $u\geqslant0$ the function $g_{\varepsilon}(u)$ is increasing.
Moreover, $g_{\varepsilon}(0)=a_{\varepsilon}=f_{\varepsilon}(x(\varepsilon))$.
This means that there exists an increasing continuous function $\varphi_{\varepsilon}:[0,r_{\varepsilon}]\to[x(\varepsilon); \pi / 4]$
such that 
\[
f_{\varepsilon}(\varphi_{\varepsilon}(u))=g_{\varepsilon}(u),\quad\varphi_{\varepsilon}(u)=f_{\varepsilon}^{-1}(g_{\varepsilon}(u)).
\]

2.  $x\in[0,x(\varepsilon)]$. In this case $f_{\varepsilon}(x)$ is monotone
decreasing and takes its values on $[f_{\varepsilon}(x(\varepsilon)),f_{\varepsilon}(0)]$:
\[
f_{\varepsilon}(0)=0,\quad  f_{\varepsilon}(x(\varepsilon))=a_{\varepsilon}<0.
\]
By (\ref{gder}), $g_{\varepsilon}(u)$ is monotone decreasing for
$u\in[-2b_{\varepsilon};0]$. Now consider the set of values of $g(u)$
for $u\in[-2b_{\varepsilon};0]$. We want to show that this set of
values contains the segment $[f_{\varepsilon}(x(\varepsilon)),f_{\varepsilon}(0)]$.
It is clear that $g(0)=a_{\varepsilon}=f_{\varepsilon}(x(\varepsilon))$.
Moreover, from representation (\ref{epsbolzerfTailorViag}) we have
\[
f_{\varepsilon}(0)=0=g_{\varepsilon}(-x(\varepsilon))+R_{\varepsilon}(0),
\]
\[
R_{\varepsilon}(0)=\frac{1}{6}\int_{0}^{x(\varepsilon)}s^{3}f^{(4)}(s)ds,\quad f^{(4)}(s)=-(1+\varepsilon)\sin s.
\]
Consequently, $R_{\varepsilon}(0)<0$ and thus $g_{\varepsilon}(-x(\varepsilon))>0$.
So, $g_{\varepsilon}(u)$ for $u\leqslant0$ reaches its maximum value
in the point $u=-2b_{\varepsilon},$ and $g_{\varepsilon}(-2b_{\varepsilon})>0$.
It follows that there exists a point $l_{\varepsilon}\in(-2b_{\varepsilon},0)$
such that $g_{\varepsilon}(l_{\varepsilon})=0$. And consequently,
there exists increasing continuous function $\varphi_{\varepsilon}:[l_{\varepsilon},0]\to[0,x(\varepsilon)]$
such that 
\[
f_{\varepsilon}(\varphi_{\varepsilon}(u))=g_{\varepsilon}(u),\quad\varphi_{\varepsilon}(u)=f_{\varepsilon}^{-1}(g_{\varepsilon}(u)).
\]

Thus, we have proved the first assertion of the lemma. Now we shall
prove inequality (\ref{bepslepsxepsabs}): 
\[
-b_{\eps}<-x(\eps)<l_{\eps}.
\]
We showed that 
\[
a_{\eps}=x(\eps)-b_{\eps}<0.
\]
It follows that $-2b_{\eps}<-b_{\eps}<-x(\eps)$. From the inequality
$g_{\eps}(-x(\eps))>0$, proved above, and definition of the function
$\varphi_{\eps}(u)$ it follows that $l_{\eps}>-x(\eps)$.

We shall use now the following notation 
\[
\phi_{\eps}(u)=\varphi_{\varepsilon}(u)-x(\varepsilon)
\]
but sometimes we will omit index $\epsilon$. Further, for simplicity
we will not write the lower index for the corresponding functions.

From representation (\ref{epsbolzerfTailorViag}) we have: 
\begin{equation}
g(u)=g(\phi(u))+R_{\eps}(\varphi(u)).\label{201906021510}
\end{equation}
The remainder term $R_{\eps}$ can be written as: 
\[
R_{\eps}(x)=\frac{1}{4!}(x-x(\eps))^{4}f_{\eps}^{(4)}(\eta_{\eps}(x))=-\frac{1+\eps}{4!}(x-x(\eps))^{4}\sin\eta_{\eps}(x)
\]
for some point $\eta_{\eps}(x)$ from the segment connecting points
$x$ and $x(\eps)$.

We subdivide the proof of the second assertion  in several parts.

1. Let us prove that for all $u\in[l_{\eps},r_{\eps}]$ the following
inequality holds: 
\begin{align}
  |u|&\leqslant|\phi(u)|\leqslant\alpha(\eps')|u|, \label{201906021420} \\
  \alpha(\eps')&=\max\Bigl\{\Bigl(1-\frac{1+\eps'}{4}\Bigr)^{-1/3}, 
  \Bigl(1-\frac{1}{3}-\frac{1+\eps'}{12}\Bigr)^{-1/2}\Bigr\}. \nonumber 
\end{align}
We will need in the proof the exact value of the constant $\alpha(\eps')$.
It is clear that for $\eps'<1$ the constant $\alpha(\eps')$ is correctly
defined and $\alpha(\eps')<2$. Since $R_{\eps}(x)\leqslant0$, 
we conclude immediately that 
\[
g(u)\leqslant g(\phi(u)).
\]
If $\phi(u)\geqslant0$ and $g(u)$ is increasing for $u\geqslant0$,
it follows that $\phi(u)\geqslant u$. Vice versa, if $\phi(u)\leqslant0$
and $g(u)$ is decreasing for $u\leqslant0$, then $\phi(u)\leqslant u$.
This proves the left inequality in (\ref{201906021420}). To check
the right inequality consider two cases:

(a) \ 
 $u\geqslant0$. As $\phi(u)\leqslant \pi / 4<1$, we have: 
\begin{align*}
g(u)&\geqslant g(\phi(u))-\frac{1+\eps'}{4!}\phi^{4}(u)\geqslant g(\phi(u))-\frac{1+\eps'}{4!}\phi^{3}(u)
\\
    &=\Bigl(\frac{1}{6}-\frac{1+\eps'}{4!}\Bigr)\phi^{3}(u)+\frac{1}{2}b_{\eps}\phi^{2}(u)+a_{\eps}\\
  &=c(\eps')\frac{1}{6}\phi^{3}(u)+\frac{1}{2}b_{\eps}\phi^{2}(u)+a_{\eps}
\\
&\geqslant c(\eps')\frac{1}{6}\phi^{3}(u)+c^{2/3}(\eps')\frac{1}{2}b_{\eps}\phi^{2}(u)+a_{\eps}=g(c^{1/3}(\eps')\phi(u)),
\end{align*}
where 
\[
c(\eps')=1-\frac{1+\eps'}{4}.
\]
For $u\geqslant0$ the function $g$ is increasing, hence 
\[
\phi(u)\leqslant c^{-1/3}(\eps')u.
\]

(b) \ $u\leqslant0$. Again using the representation for the remainder term
$R_{\eps}(x)$, we get inequalities: 
\[
g(u)\geqslant g(\phi(u))+\frac{1+\eps'}{4!}\phi^{3}(u)=\Bigl(\frac{1}{6}+\frac{1+\eps'}{4!}\Bigr)\phi^{3}(u)+\frac{1}{2}b_{\eps}\phi^{2}(u)+a_{\eps}.
\]
Since $\phi(u)=\varphi(u)-x(\eps)\geqslant-x(\eps)\geqslant-b_{\eps}$,
we have $\phi^{3}(u)=\phi(u)\phi^{2}(u)\geqslant-b_{\eps}\phi^{2}(u)$.
It follows that 
\[
g(u)\geqslant\Bigl(-\frac{1}{6}-\frac{1+\eps'}{4!}+\frac{1}{2}\Bigr)b_{\eps}\phi^{2}(u)+a_{\eps}=\tilde{c}(\eps')\frac{1}{2}b_{\eps}\phi^{2}(u)+a_{\eps},
\]
where 
\[
\tilde{c}(\eps')=-\frac{1}{3}-\frac{1+\eps'}{12}+1>\frac{1}{2},\ \mbox{for}\ \eps'<1.
\]
On the other hand, for $u\leqslant0$ we have the inequality $g(u)\leqslant\frac{1}{2}b_{\eps}u^{2}+a_{\eps}$.
We can conclude that 
\[
|\phi(u)|\leqslant\frac{1}{\sqrt{\tilde{c}(\eps')}}|u|.
\]
Inequality (\ref{201906021420}) is thus completely proved.

In this inequality put $u=r_{\eps}$. We get: 
\[
r_{\eps}\leqslant\phi(r_{\eps})=\varphi(r_{\eps})-x(\eps)=\frac{\pi}{4}-x(\eps)\leqslant\frac{\pi}{4}.
\]
On the other hand, as $x_{\eps}\leqslant \pi / 6$, we get 
\[
2r_{\eps}\geqslant\varphi(r_{\eps})-x(\eps)\geqslant\frac{\pi}{4}-\frac{\pi}{6}=\frac{\pi}{12}.
\]
Consequently, 
\[
\frac{\pi}{24}\leqslant r_{\eps}\leqslant\frac{\pi}{4}.
\]
This proves (\ref{repsineq}). Let us prove now (\ref{bepslepsxepsrel}).
Substitute $u=l_{\eps}$ to  (\ref{201906021420}). We will
get 
\[
x_{\eps}\leqslant2|l_{\eps}|.
\]
Consequently, 
\[
\frac{b_{\eps}}{|l_{\eps}|}\leqslant2\frac{b_{\eps}}{x_{\eps}}=2(1+\eps)\frac{\sin x(\eps)}{x(\eps)}\leqslant2(1+\eps')<4
\]
and (\ref{bepslepsxepsrel}) is proved.

\medskip

2.  Now using the inequality (\ref{201906021420}), for small $u$ we
will prove more exact estimate: 
\begin{equation}
|\phi(u)-u|\leqslant cu^{2},\label{201906021503}
\end{equation}
that holds for all $u\in[l_{\eps},r_{\eps}]$ with some absolute constant
$c>0$ not depending on $\eps$ and $u$. From (\ref{201906021510})
and (\ref{201906021420}) we have:: 
\begin{equation}
g(\phi(u))\leqslant g(u)+\frac{1+\eps'}{4!}\phi^{4}(u)\leqslant g(u)+cu^{4},\label{201906021525}
\end{equation}
with $c< 2^{5} / 4!$. For the second derivative of $g$ we have
\[
g''(u)=b_{\eps}+u,
\]
because for $u\geqslant-b_{\eps}$ the function $g(u)$ is downward convex.
Now we draw the tangent to $g(u)$ at the point $u$: 
\[
y(v)=g(u)+g'(u)(v-u),\ v\in\mathbb{R}.
\]
If $u\geqslant0$, then due to convexity of $g$ and due to the fact
that $\phi(u)\geqslant u$ it follows that the point $v$, where the
tangent equals $g(u)+cu^{4}$ (see (\ref{201906021525})), lies to
the right of $\phi(u)$. Then 
\[
\phi(u)\leqslant u+\frac{c}{g'(u)}u^{4}=u+\frac{cu^{4}}{b_{\eps}u+ u^{2}/2}=u+\frac{cu^{3}}{b_{\eps}+ u/2}=u+2cu^{2}.
\]
Vice versa, if $u\leqslant0$, and since $\phi(u)\geqslant-b_{\eps}$,
the same arguments as in the case $u\geqslant0$ give the inequality
\[
\phi(u)\geqslant u+\frac{c}{g'(u)}u^{4}=u+\frac{2cu^{3}}{2b_{\eps}+u}\geqslant u-2cu^{2}.
\]
This proves (\ref{201906021503}).

\medskip

3.  Let us prove the inequality for the first derivative of $\varphi(u)$.
We have: 
\[
\varphi'(u)=\frac{g'(u)}{f'(\varphi(u))}.
\]
For the derivative of $f$ we have: 
\[
f'(x)=g'(x-x(\eps))+R'(x),
\]
where 
\[
R'(x)=\frac{1}{2!}\int_{x(\varepsilon)}^{x}(x-s)^{2}f_{\varepsilon}^{(4)}(s)ds=-\frac{1+\eps}{3!}(x-x(\eps))^{3}\sin\theta_{\eps}(x)
\]
for some $\theta_{\eps}(x)$ from the segment connecting points
$x$ and $x(\eps)$. Thus, 
\begin{equation}
f'(\varphi(u))=g'(\phi(u))-\frac{1+\eps}{3!}\phi^{3}(u)\sin\theta_{\eps}(\varphi(u)).\label{201906021615}
\end{equation}

Firstly, we will show that 
\begin{equation}
\varphi'(u)\geqslant c>0\label{minophiepsbolzero}
\end{equation}
for all $u\in[l_{\eps},r_{\eps}]$ and some constant $c>0$ not depending
on $\eps$ and $u$. Consider two cases: 

(a) \ $u\geqslant0$. As $g'(u)$ is increasing, by (\ref{201906021420})
and (\ref{201906021615}) we have the inequality: 
\[
f'(\varphi(u))\leqslant g'(\phi(u))\leqslant g'(cu)\leqslant\max\{c,1\}g'(u).
\]
Consequently, as $g'(u)\geqslant0$, we get: 
\[
\varphi'(u)\geqslant\frac{g'(u)}{\max\{c,1\}g'(u)}=\frac{1}{\max\{c,1\}} .
\]

(b) $u\leqslant0$. Again using (\ref{201906021420}) and (\ref{201906021615})
we get: 
\[
f'(\varphi(u))\geqslant g'(\phi(u))\geqslant g'(cu)=\frac{c^{2}}{2}u^{2}+cb_{\eps}u\geqslant cb_{\eps}u.
\]
Using this, from $g'(u)\leqslant0$, we have: 
\[
  \varphi'(u)=\frac{-g'(u)}{-f'(\varphi(u))}\geqslant\frac{-g'(u)}{-cb_{\eps}u}=\frac{1}{c}\Bigl(1+\frac{u}{2b_{\eps}}\Bigr)\geqslant
  \frac{1}{c}\Bigl(1+\frac{-b_{\eps}}{2b_{\eps}}\Bigr)=\frac{1}{2c}.
\]
Thus, (\ref{minophiepsbolzero}) is proved.

To get the upper bound for
the second derivative $\varphi$, we will need one more inequality
for $\varphi'$. Namely, we will prove that 
\begin{equation}
|\varphi'(u)-1|\leqslant cu\label{deltaprimemajorbound}
\end{equation}
for all $u\in[l_{\eps},r_{\eps}]$ with some constant $c>0$ not depending
on $\eps$ and $u$. Due to (\ref{201906021615}) we have: 
\[
\varphi'(u)-1=\frac{g'(u)-f'(\varphi(u))}{f'(\varphi(u))}=\frac{g'(u)-g'(\varphi(u))+ ((1+\eps)/3!)\phi^{3}(u)\sin\theta_{\eps}(\varphi(u))}{f'(\varphi(u))}
\]
and 
\[
g'(\phi(u))=g'(u)+g''(u)(\phi(u)-u)+\frac{g^{(3)}(u)}{2}(\phi(u)-u)^{2}.
\]
Also by (\ref{201906021503}) 
\begin{equation}
|\varphi'(u)-1|\leqslant c\frac{g''(u)u^{2}+|u|^{3}+|u|^{4}}{|f'(\varphi(u))|} \label{deltaprimebound}
\end{equation}
for some absolute constant $c>0$, not depending on $\eps$ and $u$.
Now for $|f'(\varphi(u))|$ we will get the lower bound. Consider two
cases: 

(a) \  $u\geqslant0$. By formula (\ref{201906021615}) and inequality (\ref{201906021420})
we have 
\[
f'(\varphi(u))\geqslant g'(\phi(u))-\frac{1+\eps'}{3!}\phi^{3}(u)\geqslant g'(\phi(u))-\frac{1+\eps'}{3!}\phi^{2}(u)\geqslant c(\eps')u^{2}+b_{\eps}u,
\]
where 
\[
c(\eps')=\frac{1}{2}-\frac{1+\eps'}{3!}>\frac{1}{4}\quad  \mbox{for}\ \eps'<\frac{1}{2}.
\]
Thus we showed that 
\begin{equation}
f'(\varphi(u))\geqslant\frac{1}{4}u^{2}+b_{\eps}u,\ u\in[0,r_{\eps}].\label{fprimeepsbolzerominor}
\end{equation}
And from (\ref{deltaprimebound}) we get: 
\begin{align*}
  |\varphi'(u)-1|&\leqslant c\frac{(b_{\eps}+u)u^{2}+u^{3}+u^{4}}{u^{2}/4+b_{\eps}u}
                   =c\frac{(b_{\eps}+u)u}{u/4+b_{\eps}}+c\frac{u^{2}+u^{3}}{u/4+b_{\eps}}\\
  &\leqslant4cu+4c(u+u^{2})=c'u
\end{align*}
for some absolute constant $c'>0$, not depending on $\eps$ and $u$.

(b) \  $u\leqslant0$. Using (\ref{201906021615}), (\ref{201906021420})
and that $g'(u)$ increases we get: 
\[
f'(\varphi(u))\leqslant g'(\phi(u))-\frac{1+\eps'}{3!}\phi^{3}(u)\leqslant g'(u)+\frac{1+\eps'}{3!}\phi^{2}(u)\leqslant\tilde{c}(\eps')u^{2}+b_{\eps}u,
\]
where 
\[
\tilde{c}(\eps')=\frac{1}{2}+\alpha^{2}(\eps')\frac{1+\eps'}{3!}.
\]
Note that 
\[
\tilde{c}(0)=\frac{1}{2}+\frac{1}{6}\max\Bigl\{\Bigl(\frac{3}{4}\Bigr)^{-1/3},\Bigl(\frac{7}{12}\Bigr)^{-1/2}\Bigr\} \leqslant\frac{1}{2}+\frac{2}{6}=\frac{5}{6}.
\]
As $\tilde{c}(\eps')$ is continuous in $\eps'$, there exists $E>0$
such that for any $\eps'<E$ the following inequality holds 
\[
\tilde{c}(\eps')<\frac{1}{2}\Bigl(1+\frac{5}{6}\Bigr)=\frac{11}{12}.
\]
It follows that 
\[
f'(\varphi(u))\leqslant\frac{11}{12}u^{2}+b_{\eps}u\leqslant0
\]
and 
\begin{equation}
|f'(\varphi(u))|\geqslant\left|\frac{11}{12}u^{2}+b_{\eps}u\right|,\ u\in[l_{\eps},0].\label{fprimeepsbolzerominor2}
\end{equation}

Using this inequality in (\ref{deltaprimebound}), we get: 
\begin{align}
  |\varphi'(u)-1|& \leqslant c\frac{(b_{\eps}+u)u^{2}+|u|^{3}+|u|^{4}}{| 11u^{2}/12 +b_{\eps}u|}=c\frac{(b_{\eps}+u)|u|}{11u/12 +b_{\eps}}
                   +c\frac{|u|^{2}+|u|^{3}}{11u/12+b_{\eps}} \nonumber 
\\
&\leqslant\frac{12}{11}c|u|+12c(|u|+u^{2})\leqslant c'|u|\label{c_prime_u}
\end{align}
for some absolute constant $c'>0$ not depending on $\eps$ and $u$.
In (\ref{c_prime_u}) we used that $11u/12 +b_{\eps}\geqslant 11u/12 -u\geqslant0$.
Thus, (\ref{deltaprimemajorbound}) is completely proved.

\medskip

4. Finally, we have to prove an estimate from above for the second derivative
of $\varphi(u)$. It will be useful to introduce the notation: 
\[
\Delta(u)=\phi(u)-u=\varphi(u)-x(\eps)-u.
\]
Then inequalities (\ref{201906021503}) and (\ref{deltaprimemajorbound})
can be rewritten as follows: 
\begin{equation}
|\Delta(u)|\leqslant cu^{2},\quad|\Delta'(u)|\leqslant cu.\label{deltaineqepsbolzero}
\end{equation}
From the definition we get: 
\begin{align*}
\varphi''(u)&=\frac{g''(u)}{f'(\varphi(u))}-\frac{g'(u)f''(\varphi(u))\varphi'(u)}{(f'(\varphi(u)))^{2}}=\frac{g''(u)}{f'(\varphi(u))}-\frac{f''(\varphi(u))(\varphi'(u))^{2}}{f'(\varphi(u)))}
\\
&=\frac{g''(u)-f''(\varphi(u))(\varphi'(u))^{2}}{f'_{}(\varphi(u))}
\end{align*}
and from (\ref{epsbolzerfTailor}) 
\[
f''(x)=g''(x-x({\eps}))+R''(x),
\]
\[
R''(x)=\int_{x(\varepsilon)}^{x}(x-s)f_{\varepsilon}^{(4)}(s)ds=-\frac{1+\eps}{2}(x-x(\eps))^{2}\sin\xi_{\eps}(x)
\]
for some point $\xi_{\eps}(x)$ from the segment connecting points
$x$ and $x(\eps)$. Then 
\[
f''(\varphi(u))=g''(\phi(u))+R''(\varphi(u)).
\]
Moreover, 
\[
g''(\phi(u))=g''(u)+(\phi(u)-u)=g''(u)+\Delta.
\]
Now, from (\ref{deltaineqepsbolzero}) and (\ref{201906021420})
we get: 
\begin{align*}
|\varphi''(u)|&=\Bigl|\frac{g''(u)-g''(\varphi(u))(\varphi'(u))^{2}-R''(\varphi(u))(\varphi'(u))^{2}}{f'_{}(\varphi_{\varepsilon}(u))}\Bigr|
\\
              &\leqslant\frac{|g''(u)-g''(\varphi(u))(\Delta'+1)^{2}|}{|f'_{}(\varphi_{\varepsilon}(u))|}+\frac{cu^{2}}{|f'_{}(\varphi_{\varepsilon}(u))|}\\
  &\leqslant g''(u)\frac{|1-(\Delta'+1)^{2}|}{|f'_{}(\varphi_{\varepsilon}(u))|}+\frac{c'u^{2}}{|f'_{}(\varphi_{\varepsilon}(u))|}
\\
&\leqslant\tilde{c}\frac{g''(u)u}{|f'_{}(\varphi_{\varepsilon}(u))|}+\frac{c'u^{2}}{|f'_{}(\varphi_{\varepsilon}(u))|}
\end{align*}
for some positive constants $c,c',\tilde{c}$ not depending on $\eps$
and $u$. Again we have to consider two cases -- two signs of $u$.
If $u\geqslant0$, then by (\ref{fprimeepsbolzerominor}) we have
\[
  |\varphi''(u)|\leqslant\tilde{c}\frac{g''(u)u}{u^{2}/4+b_{\eps}u}+\frac{c'u^{2}}{u^{2}/4+b_{\eps}u}
  =\tilde{c}\frac{u+b_{\eps}}{u/4+b_{\eps}}+\frac{c'u}{u/4 +b_{\eps}}\leqslant4\tilde{c}+4c'.
\]
If $u\leqslant0$, then using (\ref{fprimeepsbolzerominor2}), we
have 
\begin{align*}
 |\varphi''(u)|  &\leqslant\tilde{c}\frac{g''(u)|u|}{|11u^{2}/12+b_{\eps}u|}+\frac{c'u^{2}}{|11u^{2}/12+b_{\eps}u|}\\
&  =\tilde{c}\frac{u+b_{\eps}}{11u/12+b_{\eps}}+\frac{c'|u|}{11u/12+b_{\eps}}\leqslant\frac{12}{11}\tilde{c}+12c'.
\end{align*}
Thus Lemma \ref{phisubstitutebolzero} is completely proved.
\end{proof}

\section{Conclusion}

We considered here only one-dimensional lattice with accent on infinite
dimension, $l_{\infty}$-initial conditions and, most important, on
the uniform boundedness. Uniform boundedness is one of the most important
stability factors for large systems. Most rigorous papers on large
systems are dedicated to equilibrium situation with Gibbs states.
However most systems are very far from equilibrium and Gibbs distribution
 could hardly play important role. One reason is as follows. It
is known that equilibrium Coulomb systems with particles of different
signs do not exist. So, to exist, the particles should move sufficiently
quickly. Unfortunately, study of such problems is now not in the list
of main directions of modern mathematics. 

There are many applied interpretations of this problem, both physical
and social. We show that for uniform boundedness,  the system
should be initially ``well organized'' -- large smooth clans with rare gaps
between them. 

It is interesting that such problems appeared to be related to fine
questions of classical mathematics, for example to Bessel functions.
And we want to thank Yu.\ Neretin for useful information concerning
Bessel functions.


\begin{thebibliography}{1}
\bibitem{fedoruck}
  {   M.V.~Fedoryuk} (1977)
\newblock  {\it Saddle Point Method}.  \newblock Nauka, Moscow (in Russian).

\bibitem{GR}
  {     I.S.~Gradshteyn and I.M.~Ryzhik} (2007)
\newblock   {\it Table of Integrals, Series, and Products.}
\newblock  Academic Press.

\bibitem{Watson}
  {  G.N.~Watson} (1922)
\newblock   {\it A Treatise on the Theory of Bessel
Functions}. \newblock Cambridge.
 

\bibitem{Luke}
  {     Y.~Luke} (1962)
\newblock   {\it Integrals of Bessel Function}.
\newblock   McGraw-Hill
Book Company.

\bibitem{LM_1}
{     A.A.~Lykov  and V.A.~Malyshev} (2020)
\newblock   Uniformly bounded initial
  chaos in large systems often intensifies infinitely.
  \newblock {\it Markov Process.\ and Relat.\ Fields\/}   {\bf 26} (2), 233--286.
\end{thebibliography}
\end{document}